\documentclass[final,twocolumn]{IEEEtran}

\usepackage{graphicx}
\usepackage{amssymb}
 
\usepackage{amsmath}
\usepackage{dsfont}
\usepackage{bm}
\usepackage{eucal}
\usepackage{cite}
\usepackage{color}
\usepackage{amsthm}

\def\R{{\mathds R}}
\def\C{{\mathds C}}

\newcommand{\be}{\begin{equation}}
\newcommand{\ee}{\end{equation}}
\newcommand{\etr}[1]{{\mathrm{etr}}\left\{#1\right\}}

\newcommand{\bzero}{{\mbox{\boldmath $0$}}}
\newcommand{\bI}{{\mbox{\boldmath $I$}}}
\newcommand{\bz}{{\mbox{\boldmath $z$}}}

\newcommand{\bx}{{\mbox{\boldmath $x$}}}

\newcommand{\bom}{{\mbox{\boldmath $m$}}}
\newcommand{\bbeta}{{\mbox{\boldmath $\beta$}}}
\newcommand{\bB}{{\mbox{\boldmath $B$}}}
\newcommand{\bA}{{\mbox{\boldmath $A$}}}
\newcommand{\bC}{{\mbox{\boldmath $C$}}}
\newcommand{\bG}{{\mbox{\boldmath $G$}}}

\newcommand{\bZ}{{\mbox{\boldmath $Z$}}}
\newcommand{\bR}{{\mbox{\boldmath $R$}}}
\newcommand{\bV}{{\mbox{\boldmath $V$}}}
\newcommand{\bE}{{\mbox{\boldmath $E$}}}
\newcommand{\bX}{{\mbox{\boldmath $X$}}}
\newcommand{\bS}{{\mbox{\boldmath $S$}}}
\newcommand{\bQ}{{\mbox{\boldmath $Q$}}}
\newcommand{\bP}{{\mbox{\boldmath $P$}}}
\newcommand{\bW}{{\mbox{\boldmath $W$}}}

\newcommand{\bH}{{\mbox{\boldmath $H$}}}

\newcommand{\tr}{\mbox{\rm Tr}\, }
\newcommand{\bGamma}{\mbox{\boldmath{$\Gamma$}}}
\newcommand{\bLambda}{\mbox{\boldmath{$\Lambda$}}}
\newcommand{\bSigma}{\mbox{\boldmath{$\Sigma$}}}
\newcommand{\bU}{{\mbox{\boldmath $U$}}}

\newcommand{\bM}{{\mbox{\boldmath $M$}}}
\newcommand{\diag}{\mbox{diag}\, }
\newcommand{\bK}{{\mbox{\boldmath $K$}}}
\newcommand{\bT}{{\mbox{\boldmath $T$}}}
\newcommand{\bD}{{\mbox{\boldmath $D$}}}
\newcommand{\dmax}{\begin{displaystyle}\max\end{displaystyle}}

\newcommand{\dmin}{\begin{displaystyle}\min\end{displaystyle}}

\newtheorem{theorem}{Theorem}
\newtheorem{lemma}{Lemma}
\newtheorem{corollary}{Corollary}

\newcommand{\test}{\mbox{$
\begin{array}{c}
\stackrel{ \stackrel{\textstyle H_1}{\textstyle >} }{ 
\stackrel{\textstyle <}{ \textstyle  H_0} }

\end{array}
$}}

\def\cC{\mbox{$\CMcal C$}}
\def\cN{\mbox{$\CMcal N$}}

\begin{document}

\title{A Unified Theory of Adaptive Subspace Detection. Part I: Detector Designs.}

\author{Danilo Orlando, \IEEEmembership{Senior Member, IEEE},  Giuseppe Ricci$^{*}$, \IEEEmembership{Senior Member, IEEE}, and Louis L. Scharf, \IEEEmembership{Life Fellow Member, IEEE}
\thanks{Danilo Orlando is with Engineering Faculty of Università degli Studi “Niccolò Cusano”, via
Don Carlo Gnocchi, 3, 00166 Roma, Italy. E-mail: {\tt danilo.orlando@unicusano.it}.}
\thanks{Giuseppe Ricci is with the Dipartimento di Ingegneria dell'Innovazione,
Universit\`{a} del Salento, Via Monteroni, 73100 Lecce, Italy.
E-Mail: {\tt giuseppe.ricci@unisalento.it}.}
\thanks{Louis L. Scharf is with the Departments of Mathematics and Statistics, Colorado
State University, Fort Collins, CO, USA
E-Mail: {\tt scharf@colostate.edu}. His work is supported by the US Office of Naval Research under contract N00014-21-1-2145, and by the US Airforce Office of Scientific Research under contract AF 9550-18-1-0087.}
\thanks{$^{*}$Corresponding author}         
}

\maketitle

\begin{abstract}
This paper addresses the problem of detecting multidimensional subspace signals, which model range-spread targets, in noise of unknown covariance. It is assumed that a primary 
channel of measurements, possibly consisting of signal plus noise, is augmented with a secondary channel of measurements 
containing only noise. The noises in these two channels share  a common covariance matrix, up to a scale, which may be known or unknown. 
The signal model is a subspace model with variations: the subspace may be known or known only by its dimension; consecutive visits to the subspace may 
be unconstrained or they may be constrained by a prior distribution. As a consequence, there are four general classes of detectors 
and, within each class, there is a detector for the case where the scale between the primary and secondary channels is known, and for the case where 
this scale is unknown.
The generalized likelihood ratio (GLR) based detectors derived in this paper, when organized with 
previously published GLR detectors, comprise a unified theory of adaptive subspace detection
from primary and secondary channels of measurements. 
\end{abstract}


\section{Introduction}

The general problem of matched and adaptive subspace detection of point-like targets in Gaussian and non-Gaussian disturbance has been addressed by  
many authors, beginning with the seminal work of Kelly and Forsythe \cite{Kelly86}, \cite{Kelly-Forsythe}. 
The  innovation of \cite{Kelly86} was to introduce a {\em homogeneous} secondary channel of signal-free measurements whose unknown covariance matrix was equal to the unknown covariance matrix of primary (or test) measurements. Likelihood theory was then used to derive what is now called the Kelly detector. In \cite{Kelly-Forsythe}, the adaptive subspace detection was formulated in terms of the so-called
generalized multivariate analysis of variance for complex variables.
These papers were followed by the important 
adaptive detectors of \cite{Chen-Reed}, \cite{Robey}. Then in 1995 and 1996, a scale-invariant {\em adaptive subspace detector}, now commonly called 
ACE (adaptive coherence estimator), was introduced. In \cite {CLR1995} this detector was derived as an asymptotic approximation to the generalized 
likelihood ratio (GLR) to detect a coherent signal in compound-Gaussian noise with known spectral properties, and in \cite{Scharf-McWhorter} it 
was derived as an {\em estimate and 
plug} version of the scale-invariant {\em matched subspace detector} \cite{Scharf-book},  \cite{Scharf-Friedlander1994}. 
Interestingly, in \cite {Kraut-Scharf1999} the authors showed that ACE was a likelihood ratio detector for a {\em non-homogeneous} secondary channel of measurements whose unknown covariance matrix was a scaled version of the unknown covariance matrix of the primary channel.  The scale was unknown. Then, in \cite{Kraut-Scharf-McWhorter} it was shown that ACE is a uniformly most powerful invariant (UMPI) detector. In subsequent years there has been a flood of important papers. Among published references on adaptive detection we cite here 
\cite[and references therein]{Bose-Steinhardt, Gerlach1997, CDMR2001, Gini-Farina2002,BDMGR,BBORS2007,CDMO2016_1,CDMO2016_2,CFR2020}. 
All of this work is addressed to adaptive detection in what might be called a first-order model for measurements. That is, the measurements under test may contain a signal in a known subspace embedded in Gaussian noise of unknown covariance, but no prior distribution is assigned to the location of the signal in the subspace. 
In particular, in \cite{CDMR2001,BDMGR} the authors extend adaptive subspace detection to range-spread targets deriving likelihood ratio detectors that were then compared to estimate and plug adaptations.  
The first attempt to replace this model by a second-order model was made in \cite{Ricci-Scharf}, where the authors used a Gaussian model for the signal. The covariance matrix for the signal was constrained by a known subspace model. The resulting {\em second-order matched subspace detector} was derived \cite{Ricci-Scharf}, and an estimate and plug adaptation from secondary measurements was proposed. 

The aim of the current paper is to extend the results of \cite{CDMR2001,Ricci-Scharf,BDMGR,BCCRV} to include all variations on adaptive subspace detection in first- and second-order models for a subspace signal to be detected. These models include 
signals that lie in a known subspace or in an unknown subspace of known dimension. They will be clarified in due course.

Our results are motivated by the problem of 
detecting range-spread targets from an active radar system. 
As a matter of fact, targets may  
be resolved into a number of scattering centers
depending on the range extent of the target, the range resolution  
of the radar, 
and its operating frequency.
Measurements indicate that radar properties of several targets  
may be modeled in terms of a set of scattering 
centers each parameterized by its range, amplitude, and, possibly, polarization elipse \cite{Steedly-Moses}.
However, our 
framework and corresponding 
results are actually much more generally applicable, as they apply to sonar, communications, hyperspectral imaging, and to any problem where a signal to be detected lies in a known subspace, or in a subspace known only by its dimension. 
In all of these applications, measurements in a primary channel may contain signal  plus Gaussian noise. Measurements in a secondary channel contain only noise. The noises in the two channels are independent, but they share a common covariance matrix, at least to within 
an unknown scale. The case of a common covariance matrix in the two channels is typically referred to as a case of {\em homogeneous environment}, 
while the more general case of an unknown scale factor is commonly referred to as a case of {\em partially-homogeneous environment}.
As for the signal components, they are determined by a visit to a subspace. According to a {\em first-order model} for these visits, 
there is no constraint on their location in the subspace; as a consequence the subspace signal model modulates the mean of a multivariate Gaussian 
distribution. According to a {\em second-order model}, the location in the subspace is ruled by a prior distribution, which is taken to be a Gaussian distribution; as a consequence the subspace signal model modulates the covariance matrix  of a multivariate Gaussian distribution. For each of 
these variations on the problem of adaptively detecting a subspace signal, we derive a 
detector based upon the GLR, or generalized likelihood ratio test (GLRT)
(for the definition of GLRT see \cite{VanTrees-book1}).
Recall that the GLRT compares a GLR statistic to a threshold $\eta$,
set according to the desired probability of false alarm ($P_{fa}$), to discriminate between the noise-only hypothesis ($H_0$)
and the signal-plus-noise hypothesis ($H_1$). Hereafter, $\eta$ will denote any modification of the original threshold.
Taken together, our results comprise a unified theory of adaptive subspace detection.

\subsection{A preview of the paper}
Before proceeding with the derivations, we summarize below the different variations on a multidimensional subspace signal model
addressed in this paper:
\begin{itemize}
\item 
The signal visits a  {\em known subspace, unconstrained by a prior distribution.}  We call this a {\em first-order model}, as the signal appears as a {\em low-rank component} in the mean of a multivariate Gaussian distribution for the measurements.  
When there is only one  measurement in the primary channel, 
then the GLRTs are those of \cite{Kelly86,CLR1995,Kraut-Scharf1999}. For multiple measurements these results are extended in \cite{CDMR2001,BDMGR}. 
Herein, we assume the signal belongs to a subspace as in \cite{BDMGR}, but we do not assume the presence 
of structured interferers. These cases are reviewed only, as they form the basis of our extensions to other models.
\item 
The signal visits an  {\em unknown subspace of known dimension,  unconstrained by a prior distribution.}  Again we call this a {\em first-order model}. The derived
GLRTs are original.
\item 
The signal visits a  {\em known subspace,  constrained by 
a Gaussian prior distribution.}  We call this a 
{\em second-order  model}, as the signal model appears 
as a {\em constrained, low-rank  component} in 
the covariance matrix  of a multivariate Gaussian distribution 
for the measurements.  
Adaptive estimate and plug  GLRTs 
have been derived in \cite{Ricci-Scharf,Jin-Friedlander}. 
The GLRTs of the current paper for this problem are original.
\item
The signal visits an  {\em unknown subspace of known dimension,  
constrained by a Gaussian prior distribution};
this is a {\em second-order  model}.
The estimated low-rank covariance matrix for the subspace signal may be 
called an adaptive factor model. 
The resulting GLRTs appear in this paper for the first time.
\end{itemize}

\subsection{Notation}
In the sequel, vectors and matrices are denoted by boldface lower-case and upper-case letters, respectively.
Symbols $\det(\cdot)$, $\tr(\cdot)$, 
$\etr \cdot$, 
$(\cdot)^T$, $(\cdot)^*$, and $(\cdot)^\dag$ denote the determinant, trace, exponential of the trace, transpose, complex conjugate,
and conjugate transpose, respectively. 
As to numerical sets, 
$\C$ is the set of 
complex numbers, $\C^{N\times M}$ is the Euclidean space of $(N\times M)$-dimensional 
complex matrices, and $\C^{N}$ is the Euclidean space of $N$-dimensional 
complex vectors. 
$\bI_n$ and $\bzero_{m,n}$ stand for the $n \times n$ identity matrix and the $m \times n$ null matrix.
$\langle \bH \rangle$ denotes the space spanned by the columns of the matrix $\bH
\in \C^{N\times M}$.
Given $a_1, \ldots, a_N \in\C$, $\diag(a_1, \ldots, a_N)
\in\C^{N\times N}$ indicates 
the diagonal matrix whose $i$th diagonal element is $a_i$.

We write $\bz\sim \cC\cN_N(\bx, \bSigma)$ to say that the $N$-dimensional random vector 
$\bz$ is  a complex normal random vector with mean vector $\bx$ and covariance matrix $\bSigma$. 
Moreover, $\bZ=[\bz_1 \cdots \bz_K] \sim \cC\cN_{NK}(\bX, \bI_K\otimes \bSigma)$, with $\otimes$ denoting Kronecker product and
$\bX=[\bx_1 \cdots \bx_K]$,
means that
$\bz_k\sim \cC\cN_N(\bx_k, \bSigma)$ and the columns of $\bZ$ are statistically independent.
The acronym PDF stands for probability density function, which is generally denoted $f(\bz; \Theta)$, where $\Theta$ denotes the set of parameters that determines the PDF. When it is necessary to speak of a likelihood function, then the roles of $\bz$ and $\Theta$ are reversed so that $\ell(\Theta;\bz)$ denotes the likelihood of parameters $\Theta$, given the measurement $\bz$, and
$L(\Theta;\bz)=\log \ell(\Theta;\bz)$ is the log-likelihood. $E[\cdot]$ denotes the statistical expectation.
Finally, $\makebox{vec}(\cdot)$ is the column vectorizing operator.

\section{Four Problems in Adaptive Subspace Detection}

For subsequent developments, let us denote by  $\bZ_P=[ \bz_1 \cdots \bz_{K_P}]\in\C^{N \times K_P}$ the matrix 
of the measurements in the primary channel and by 
$\bZ_S=[ \bz_{K_P+1} \cdots \bz_{K_P+K_S}]\in\C^{N \times K_S}$ the matrix of the measurements in the secondary channel. 
In a radar problem 
the measurements are $N$-dimensional vectors of space-time samples: the radar system transmits a burst of $N_p$ radio frequency (RF) pulses and the baseband representations of the  RF signals collected at the radar antenna are sampled to
form range-gate samples for each pulse.
If the signal presence is sought in a subset of $K_P$
range gates, the primary channel consists of
$N_a N_p K_P$ samples.
The samples corresponding to any range gate are arranged in a column vector $\bz_k\in\C^{N}$ 
with $N = N_aN_p$. The secondary channel consists of the
outputs of $K_S$ properly selected range gates  \cite{RichardsBasicPrinciples}. Finally, let 
$\bZ=[\bZ_P \ \bZ_S] \in \C^{N \times K}$ be the overall data matrix with $K=K_P+K_S$.

\subsection{First-order models}

In a first-order model for measurements, the adaptive 
detection problem 
may be  formulated as the following
test of hypothesis $H_0$ vs alternative $H_1$:
\begin{equation}
\begin{array}{ll}
H_{0}: & 
\left\{
\begin{array}{l}
\bZ_P \sim \cC \cN_{NK_P} (\bzero_{N,K_P}, \bI_{K_P}\otimes \bR) \\ \bZ_S \sim \cC\cN_{NK_S} (\bzero_{N,K_S},  \bI_{K_S}\otimes \gamma\bR)
\end{array} \right. \\ \\
H_{1}: &  
\left\{
\begin{array}{l}
\bZ_{P} \sim \cC\cN_{NK_P} (\bH \bX, \bI_{K_P}\otimes\bR) 
\\ 
\bZ_S \sim {\cC\cN}_{NK_S} (\bzero_{N,K_S}, \bI_{K_S}\otimes\gamma \bR)
\end{array} \right.
\end{array} 
\label{FO-HT}
\end{equation}
where 
$\bH \in \C^{N \times r}$ is either a known matrix or 
an unknown matrix with known rank $r$, $r \leq N$, $\bX=[\bx_1 \cdots \bx_{K_P}] \in \C^{r \times K_P}$ is the matrix of the unknown signal coordinates,  
$\bR \in \C^{N \times N}$ is an unknown positive definite
covariance matrix while $\gamma>0$ is either a known or an unknown parameter. 
In the following, we suppose that $K_S \geq N$ and, without loss of generality, that
$\bH$ is a slice of unitary matrix.

\subsection{Second-order models}
In a second-order model for measurements, the distributions above are treated as conditional distributions, 
and a prior Gaussian distribution is assumed for the matrix $\bX$, namely, $\bX\sim \cC\cN_{NK_P}(\bzero, \bI_{K_P} \otimes \bR_s)$.
The joint distribution of $\bZ_P$ and $\bX$ is marginalized for $\bZ_P$ obtaining that $\bZ_P\sim \cC\cN_{NK_P}(\bzero, \bI_{K_P} \otimes (\bH \bR_s \bH^{\dag}+\bR))$. 

The adaptive detection problem 
may be  formulated as the following
test of hypothesis $H_0$ vs alternative $H_1$:
\begin{equation}
\begin{array}{ll}
H_{0}: & \!\!\! \!\!
\left\{\!
\begin{array}{l}
\bZ_P \sim \cC \cN_{NK_P} (\bzero_{N,K_P}, \bI_{K_P}\otimes \bR) \\ \bZ_S \sim \cC\cN_{NK_S} (\bzero_{N,K_S},  \bI_{K_S}\otimes \gamma\bR) 
\end{array} \right.
\\ \\
H_{1}: & \!\!\!\!\! 
\left\{\!
\begin{array}{l}
\bZ_{P} \sim \cC\cN_{NK_P} (\bzero_{N,K_P}, \bI_{K_P}\otimes (\bH \bR_s \bH^{\dag}+\bR)) \\
\bZ_S \sim {\cC\cN}_{NK_S} (\bzero_{N,K_S}, \bI_{K_S}\otimes\gamma \bR)
\end{array} \right.
\end{array} 
\label{SO-HT}
\end{equation}
where
$\bH \in \C^{N \times r}$ is either a known matrix or an 
unknown matrix with known rank $r$, $r \leq N$,   
$\bR_s \in \C^{r \times r}$ is an unknown positive semidefinite matrix (in order to account for possible correlated sources),
$\bR \in \C^{N \times N}$ is an unknown positive definite
matrix, while $\gamma>0$ is either a known or an unknown parameter. 
Again, we suppose that $\bH$ is a slice of unitary matrix and $K_S \geq N$.

\subsection{Interpretations and important statistics}
In the derivation of adaptive subspace detectors for first-order models, several data matrices and derived statistics arise. They are summarized and annotated here. 
\begin{itemize}
\item $\bS_S=\bZ_S\bZ_S^{\dag} \in \C^{N \times N}$: sample covariance matrix for secondary channel; for $K_S\ge N$, the  covariance matrix $\bS_S$ is  positive definite with probability (wp) $1$;
\item $\bS_P=\bZ_P\bZ_P^{\dag} \in \C^{N \times N}$: sample covariance matrix for primary channel; $S_P$ is positive semidefinite with rank $\min(K_P,N)$ wp $1$;
\item $ \bT_P=\bS_S^{-1/2}\bZ_P\bZ_P^{\dag}\bS_S^{-1/2}\in \C^{N \times N}$: sample covariance matrix for  measurements in the primary channel that have been whitened by  the square root of the sample covariance matrix computed in the secondary channel; $\bT_P$ is positive semidefinite of rank $\min(K_p,N)$; 
\item $\bG=\bS_S^{-1/2}\bH \in \C^{N \times r}$: {\em whitened} subspace basis; $\bH \in \C^{N\times r}$ is a unitary basis for the $r$-dimensional subspace $\langle \bH \rangle$;
\item $\bP_G^{\perp}=\bI_N-\bG(\bG^{\dag}\bG)^{-1}\bG^{\dag} \in \C^{N\times N}$: projection matrix  onto the orthogonal complement of the dimension-$r$ subspace $\langle \bG\rangle$.
\end{itemize}
Importantly, the eigenvalues of the statistics $\bT_P$  and $\bP_G^{\perp}\bT_P\bP_G^{\perp}$ are two dramatic compressions of the primary and secondary data that  figure prominently in the first-order detectors to be derived in this paper.

\section{First-order detectors: derivations}

The GLRTs for problem (\ref{FO-HT}) can be obtained by exploiting the results
in \cite{BDMGR}. 
Therein, both homogeneous and partially-homogeneous environments are considered, and measurements
contain noise plus interference drawn 
from a subspace that is either known or unknown up to its rank. As a matter of fact, the derivation of the compressed likelihood under the 
$H_0$ hypothesis in \cite{BDMGR} is the starting point for the derivation of the GLRTs for problem (\ref{FO-HT}).
\medskip

The joint PDF of 
primary and secondary data
is given by
\begin{multline*}
f_1( \bZ; \bR,  \bX, \bH, \gamma) =
\frac{\etr{-\frac{1}{\gamma} \bR^{-1}\bZ_S \bZ_S^{\dag}}}
{\pi^{N K} \gamma^{NK_S} \det^{K}(\bR)} 
\\  
\times
\etr{-
\bR^{-1}   
\left( \bZ_P - \bH \bX \right) \left( \bZ_P - \bH \bX \right)^{\dag}}
\end{multline*}
under $H_1$ and under $H_0$ by
\be
f_0( \bZ; \bR, \gamma) =
\frac{\etr{-\bR^{-1} \bZ_P \bZ_P^{\dag}+\frac{1}{\gamma} \bR^{-1} \bZ_S \bZ_S^{\dag}}}
{\pi^{N K} \gamma^{NK_S} \det^{K}(\bR)}.
\label{eq:PDF_H0}
\ee

\subsection{Known subspace $\langle \bH \rangle$, known  $\gamma$
}

Under $H_1$, the likelihood is maximized through the maximum likelihood (ML) estimates of $\bR$ and $\bX$ to produce the 
partially-compressed likelihood \cite{BDMGR}
\begin{align}
\label{eq:PDF_H1_FO_knownH}
&\ell_1(\widehat{\bR}, \widehat{\bX}, \bH, \gamma; \bZ) =
\left(\frac{K}{e \pi}\right)^{N K} \frac{1}{\gamma^{K_P(K-N)}}
\frac{1}{\det^{K}(\bS_S)} 
\\ \nonumber &\times
\frac{1}{ {\det}^K \left[ \frac{1}{\gamma}  \bI_{K_P} +
\left( \bS_S^{-1/2} \bZ_P \right)^{\dag} \bP_G^{\perp}
\left( \bS_S^{-1/2} \bZ_P \right)
\right] }
\\ \nonumber &=
\left(\frac{K}{e \pi}\right)^{N K} \frac{1}{\gamma^{K_SN}}
\frac{1}{\det^{K}(\bS_S)}  \nonumber
\\  \nonumber
&\times
\frac{1}{ {\det}^K \left[ \frac{1}{\gamma}  \bI_{N} +
\bP_G^{\perp} \left( \bS_S^{-1/2} \bZ_P \right)  
\left( \bS_S^{-1/2} \bZ_P \right)^{\dag} \bP_G^{\perp}
\right] }
\end{align}
where 
we have used the identity
${\det} \left[ \frac{1}{\gamma}  \bI_{M} +
\bA \bB
\right]= \gamma^{N-M}
{\det} \left[ \frac{1}{\gamma}  \bI_{N} +
\bB \bA
\right]$
with $\bA \in \C^{M \times N}$ and $\bB \in \C^{N \times M}$.
It is also straightforward to show that compressed likelihood under $H_0$ is
\begin{align}
\label{eq:PDF_H0_FO_knownH}
&\ell_0( \widehat{\bR}, \gamma; \bZ) =
\left(\frac{K}{e \pi}\right)^{N K} 
\frac{\gamma^{-K_P(K-N)}\det^{-K}(\bS_S)}{ {\det}^K \left[ \frac{1}{\gamma}  \bI_{K_P} + \bZ_P^{\dag}
\bS_S^{-1} \bZ_P 
\right] }\nonumber
\\ &=
\left(\frac{K}{e \pi}\right)^{N K} 
\frac{\gamma^{-K_SN} \det^{-K}(\bS_S)}{ {\det}^K \left[ \frac{1}{\gamma}  \bI_{N} + 
\bS_S^{-1/2} \bZ_P \bZ_P^{\dag} \bS_S^{-1/2}
\right] }.
\end{align}
\medskip

It follows that the GLRT for homogeneous environment (i.e., $\gamma=1$) and $r < N$, referred to in the following as
first-order known subspace in homogeneous environment (FO-KS-HE) detector, is given by
\be
\label{eq:1S-FO-GLRT-KH-HE}
\frac{{\det}\left[ \bI_{K_P} + \bZ_P^{\dag}
\bS_S^{-1} \bZ_P \right]}{
{\det} \left[ \bI_{K_P} +
\left( \bS_S^{-1/2} \bZ_P \right)^{\dag} \bP_G^{\perp}
\left( \bS_S^{-1/2} \bZ_P \right)
\right]
}
\test \eta
\ee
or, equivalently, as
\be
\label{eq:1S-FO-GLRT-KH-HE_alternative_form}
\frac{{\det}\left[ \bI_{N} + \bT_P \right]
}{
{\det} \left[ \bI_{N} +
\bP_G^{\perp} \bT_P \bP_G^{\perp}
\right]
}
\test \eta.
\ee
The expression  in eqn \eqref{eq:1S-FO-GLRT-KH-HE_alternative_form}  illuminates the role of the {\em secondarily whitened primary data} $\bS_S^{-1/2}\bZ_P$,  its corresponding  whitened sample covariance $\bT_P$,  and the sample covariance of whitened measurements after their projection onto the subspace $\bP_G^{\perp}$. The GLRT is a function only of the eigenvalues of $\bT_P$ and the eigenvalues of $\bP_G^{\perp} \bT_P\bP_G^{\perp}$.
For $r=N$ and $\gamma=1$ the GLRT reduces to
$$
{\det}\left[ \bI_{K_P} + \bZ_P^{\dag}
\bS_S^{-1} \bZ_P \right]
\test \eta
\label{eq:1S-FO-GLRT-KH-HE-r=N}
$$
or equivalently to
$$
{\det}\left[ \bI_{N} + \bT_P
 \right]
\test \eta.
$$
These GLRTs are derived  for $\gamma=1$, but generalization to any known value of $\gamma$ is obviously straightforward.
In particular, if $\gamma$ is known we can normalize the secondary data by the square root of $\gamma$, thus obtaining the homogeneous environment. For this reason herafter  we will focus on $\gamma=1$ if $\gamma$ is known.
\medskip

\subsection{Known subspace $\langle \bH \rangle$, unknown  $\gamma$}

Determining the GLRT for a partially-homogeneous environment
requires one more maximization of the likelihoods with respect to $\gamma$, namely the
computation of 
\[
\max_{\gamma>0} \ell_1(\widehat{\bR}, \widehat{\bX}, \bH, \gamma; \bZ)
\quad \mbox{and}\quad
\max_{\gamma>0} \ell_0(\widehat{\bR}, \gamma; \bZ).
\]

For $r=N$ the likelihood under $H_1$ is unbounded with respect to $\gamma >0$ and, hence, the GLRT does not exist.
Therefore we assume $r<N$.
The following result derived in \cite{CDMR2001,BDMGR} is recalled here for the sake of completeness.
 
\begin{theorem}
\label{theorem_minimization_over_gamma_FH_PHE}
Let $\bM \in \C^{K_P \times K_P}$ be a positive semidefinite (Hermitian) matrix of 
rank $t$ ($1 \leq t \leq K_P$). Then, the function
\be
f(\gamma)=\gamma^{\frac{K_P(K-N)}{K}} \det \left( \frac{1}{\gamma}
\bI_{K_P} + \bM \right), \quad \gamma >0,
\label{f_gamma}
\ee
attains its absolute minimum at the unique positive solution of
\be
\sum_{k=K_P-t+1}^{K_P}
\frac{\lambda_k \gamma}{\lambda_k \gamma +1}=
\frac{NK_P}{K}
\label{eq_min_f_gamma}
\ee
where the $\lambda_k$s are the eigenvalues of the matrix $\bM$ arranged in increasing order ($\lambda_k=0$, $k=1, \ldots, K_P-t$) and
provided that $t > \frac{NK_P}{K}$.
If $t = \frac{NK_P}{K}$, then $f(\gamma)$ does not possess the absolute minimum over $(0, +\infty)$, but its infimum is positive; finally, if
$t < \frac{NK_P}{K}$, the infimum of $f(\gamma)$ over $(0, +\infty)$ is zero.
\end{theorem}

\begin{proof}
See \cite{CDMR2001,BDMGR}.
\end{proof}
 
To use this theorem, it is necessary to determine the rank of the matrices
$\bM_0=\bZ_P^{\dag}
\bS_S^{-1} \bZ_P$ and
$\bM_1=\left( \bS_S^{-1/2} \bZ_P \right)^{\dag} \bP_G^{\perp}
\left( \bS_S^{-1/2} \bZ_P \right)$ and whether or not  the 
condition on the rank is satisfied. Preliminarily,
we give the following lemma that can be easily proved following the lead of \cite[Theorem 3.1.4 pag. 82]{Muirhead}.
\bigskip

\begin{lemma}
Let $\bz_1, \ldots, \bz_m$ be $m$ independent and complex normal Gaussian vectors
with positive definite covariance matrix, i.e.,
$\bz_k \sim \cC\cN_N({\bom}_k, {\bR}_k)$. The rank of the matrix $[\bz_1 \cdots \bz_m]$ is equal to the minimum among $m$ and $N$ with probability one.
\end{lemma}
\bigskip

It is also easy to prove
the following theorem concerning the rank of the matrices $\bM_0$ and $\bM_1$.
\bigskip

\begin{theorem}
\label{TheoremRank02}
The rank of 
\be
\bM_0=\bZ_P^{\dag}
\bS_S^{-1} \bZ_P
\label{def:M0}
\ee 
is
$t_0=\min(K_P,N)$
and $t_0 > \frac{NK_P}{K}$ since $K>K_P$ and $K > N$.
Similarly, 
the rank of 
\be
\bM_1=\left( \bS_S^{-1/2} \bZ_P \right)^{\dag} 
\bP_G^{\perp}
\left( \bS_S^{-1/2} \bZ_P \right)
\label{def:M1}
\ee
is
$t_1=\min(K_P,N-r)$. It follows that
$t_1>NK_P/K$ when $K_P\leq N-r$ (since $K>N$); for $N-r < K_P$, the condition is $N-r>NK_P/K$, which requires $r<N(1-K_P/K)$. 


\end{theorem}
\smallskip


It follows that, under the condition $t_1 > \frac{NK_P}{K}$, the GLRT for partially-homogeneous environment, referred to in the following as 
first-order known subspace in partially-homogeneous environment 
(FO-KS-PHE) detector, is given by
\be
\frac{\widehat{\gamma}_0^{\frac{K_P(K-N)}{K}} {\det}\left[ \frac{1}{\widehat{\gamma}_0}  \bI_{K_P} + \bM_0 \right]
}{\widehat{\gamma}_1^{\frac{K_P(K-N)}{K}}
{\det} \left[ \frac{1}{\widehat{\gamma}_1}  \bI_{K_P} +
\bM_1
\right]
}
\test \eta
\label{eq:1S-FO-GLRT-KH-PHE}
\ee
where $\widehat{\gamma}_i$, $i=0,1,$ can be computed using 
Theorem \ref{theorem_minimization_over_gamma_FH_PHE} and $\bM_0$ and $\bM_1$ are given by eq. (\ref{def:M0}) and
eq. (\ref{def:M1}), respectively.
The equivalent form is more illuminating:
$$
\frac{\widehat{\gamma}_0^{N(1-K_P/K)} {\det}\left[ \frac{1}{\widehat{\gamma}_0}  \bI_{N} + \bT_P \right]
}{\widehat{\gamma}_1^{N(1-K_P/K)}
{\det} \left[ \frac{1}{\widehat{\gamma}_1}  \bI_{N} +
\bP_G^{\perp}\bT_P\bP_G^{\perp}
\right]
}
\test \eta.
$$ 
 Again, the GLRT is a function only of the eigenvalues of 
$\bT_P$ and $\bP_G^{\perp}\bT_P\bP_G^{\perp}$. In fact, $\bM_0$ and $\bT_P$ share the nonzero eigenvalues. Similarly for $\bM_1$ and $\bP_G^{\perp}\bT_P\bP_G^{\perp}$.

\subsection{Unknown subspace $\langle \bH \rangle$ of known dimension, known  $\gamma$}

The signal subspace $\langle \bH \rangle$ is unknown, but its rank $r \leq N$ is known.
To compute the compressed likelihood under $H_1$, 
the parameter $\bH$ is replaced by its ML estimate in
eq. (\ref{eq:PDF_H1_FO_knownH}). The maximization with respect to $\bH$ can be conducted as shown in \cite{BDMGR}. The result is
\be
\ell_1(\widehat{\bR}, \widehat{\bX}, \widehat{\bH}, \gamma; \bZ) =
\frac{[K/(e\pi)]^{NK}}{\gamma^{K_P(K-N)}}
\frac{1}{\det^{K}(\bS_S)} 
\frac{1}{ g_1^K(\gamma) }
\label{eq:PDF_H1_FO_unknownH}
\ee
where
$$
g_1(\gamma)=
\left\{
\begin{array}{ll}
\gamma^{N-r-K_P} \prod_{i=1}^{N-r} \left( \frac{1}{\gamma} + \sigma^2_i
\right), & m_1 \geq r+1 \\
\left( \frac{1}{\gamma}\right)^{K_P}, & \mbox{otherwise}
\end{array}
\right.
$$
with
$m_1=\min(N,K_P)$ the rank of the matrix $\bS_S^{-1/2} \bZ_P \bZ_P^{\dag}
\bS_S^{-1/2}$ and $\sigma^2_i$, $i=1, \ldots, N,$ the corresponding eigenvalues arranged in increasing order. Moreover, the compressed likelihood under $H_0$ can be re-written as
\be
\ell_0(\widehat{\bR}, \gamma;  \bZ) =
\frac{[K/(e\pi)]^{NK}}{\gamma^{K_P(K-N)}}
\frac{1}{\det^{K}(\bS_S)} 
\frac{1}{ g_0^K(\gamma) }
\label{eq:PDF_H0_FO_unknownH}
\ee
with $g_0(\gamma)= \gamma^{N-K_P} \prod_{i=1}^{N} \left( \frac{1}{\gamma} + \sigma^2_i  \right)$.

\smallskip

It follows that, if $\min(N,K_P) \geq r +1$, the GLRT for homogeneous environment, referred to in the following as 
first-order unknown subspace in homogeneous environment 
(FO-US-HE) detector, is given by
\be
\frac{ \prod_{i=1}^{N} \left( 1 + \sigma^2_i
 \right)
}{
\prod_{i=1}^{N-r} \left( 1 + \sigma^2_i
 \right)
}
= \prod_{i=N-r+1}^{N} \left( 1 + \sigma^2_i
 \right)
\test \eta.
\label{eq:1S-FO-GLRT-UH-HE}
\ee
For $m_1< r +1$, the GLRT reduces to
\be
\prod_{i=1}^{N} \left( 1 + \sigma^2_i
 \right)= {\det}\left( \bI_N + \bS_S^{-1/2} \bZ_P \bZ_P^{\dag}
 \bS_S^{-1/2}\right)
\test \eta.
\label{eq:1S-FO-GLRT-UH-HE-r=N}
\ee
Notice also that condition $m_1 < r +1$ is equivalent to $N=r$ if $N < K_P$ (recall that $N \geq r$) or to $K_P < r+1$ if $K_P \leq N$.
\medskip

\subsection{Unknown subspace $\langle \bH \rangle$ of known dimension, unknown  $\gamma$}

To obtain the GLRT for partially-homogeneous environment
we have to maximize the partially-compressed likelihoods over $\gamma$.
We focus on $m_1\geq r +1$; in fact, for $m_1 < r +1$ the likelihood under $H_1$ is unbounded with respect to $\gamma$ and, hence, the GLRT does not exist. 
Equivalently, we have to minimize with respect to $\gamma$ the following functions
\[
f_1(\gamma) = \gamma^{\frac{-K_PN}{K}} \prod_{i=N-m_1+1}^{N-r} \left( 1 + \gamma \sigma^2_i \right)
\]
and 
\[
f_2(\gamma) = \gamma^{\frac{-K_PN}{K}} \prod_{i=N-m_1+1}^{N} \left( 1 + \gamma \sigma^2_i \right).
\]
Proceeding as in the proof of Theorem \ref{theorem_minimization_over_gamma_FH_PHE}, we obtain the following results.
 
 \begin{corollary}
 \label{corollary1_minimization_over_gamma_FH_PHE}
 The function
 $$
 f_1(\gamma)=\gamma^{\frac{-K_PN}{K}}
 \prod_{i=N-m_1+1}^{N-r} \left( 1 + \gamma \sigma^2_i \right)
 $$
 attains its absolute minimum over $(0,+\infty)$
 at the unique positive solution of
\be
\sum_{i=N-m_1+1}^{N-r}
\frac{\sigma^2_i \gamma}{\sigma^2_i \gamma +1}=
\frac{NK_P}{K}
\label{eq_min_f_1_gamma}
\ee
provided that $m_1-r > \frac{NK_P}{K}$.
If $m_1-r  = \frac{NK_P}{K}$, then $f_1(\gamma)$ does not possess the absolute minimum over $(0, +\infty)$, but its infimum is positive; finally, if
$m_1-r  < \frac{NK_P}{K}$, the infimum of $f_1(\gamma)$ over $(0, +\infty)$ is zero.
\end{corollary}
 \bigskip
 
 \begin{corollary}
 \label{corollary2_minimization_over_gamma_FH_PHE}
 The function
 $$
 f_2(\gamma)=\gamma^{\frac{-K_PN}{K}}
 \prod_{i=N-m_1+1}^{N} \left( 1 + \gamma \sigma^2_i \right)
 $$
 attains its absolute minimum over $(0,+\infty)$
 at the unique positive solution of
\be
\sum_{i=N-m_1+1}^{N}
\frac{\sigma^2_i \gamma}{\sigma^2_i \gamma +1}=
\frac{NK_P}{K}
\label{eq_min_f_2_gamma}
\ee
provided that $m_1 > \frac{NK_P}{K}$.
If $m_1 = \frac{NK_P}{K}$, then $f_2(\gamma)$ does not possess the absolute minimum over $(0, +\infty)$, but its infimum is positive; finally, if
$m_1 < \frac{NK_P}{K}$, the infimum of $f_2(\gamma)$ over $(0, +\infty)$ is zero.
\end{corollary}
\bigskip

It follows that under the condition $m_1 > \frac{NK_P}{K}+r$
the GLRT, referred to in the following as 
first-order unknown subspace in partially-homogeneous environment 
(FO-US-PHE) detector, can be written as 
\be
\frac{\widehat{\gamma}_0^{N\left(1-\frac{K_P}{K}\right)} 
\prod_{i=1}^{N} \left( \frac{1}{\widehat{\gamma}_0} + \sigma^2_i  \right)
}{
\widehat{\gamma}_1^{N\left(1-\frac{K_P}{K}\right)-r} 
\prod_{i=1}^{N-r} \left( \frac{1}{\widehat{\gamma}_1} + \sigma^2_i  \right)
}
\test \eta
\label{eq:1S-FO-GLRT-UH-PHE}
\ee
where $\widehat{\gamma}_0$ and $\widehat{\gamma}_1$ can be computed using 
Corollary \ref{corollary1_minimization_over_gamma_FH_PHE}
and \ref{corollary2_minimization_over_gamma_FH_PHE}, respectively. The detector is a function of the eigenvalues of the statistic $\bT_P$.

\section{Second-order detectors: derivations}
The joint PDF of primary and secondary data is given by
\begin{align}
f_1( \bZ ; \bR, \bR_s, \bH, \gamma) &=
\frac{\etr{-\frac{1}{\gamma} \bR^{-1} 
\bZ_S \bZ_S^{\dag}}}{\pi^{NK} \gamma^{NK_S}} \nonumber
\\ &\times
\frac{\etr{-\left(\bH \bR_s \bH^{\dag} + \bR \right)^{-1}   
\bZ_P \bZ_P^{\dag}}}
{\det^{K_P}(\bH \bR_s \bH^\dag + \bR)\det^{K_S}(\bR)}
\label{eqn:f1_SO}
\end{align}
under $H_1$ and is expressed by eq. (\ref{eq:PDF_H0}) under $H_0$.
The compressed likelihood under $H_0$  has already been computed to implement the GLRTs for first-order models.

\subsection{Unknown subspace $\langle \bH \rangle$ of known dimension, known  $\gamma$}
In eq. \eqref{eqn:f1_SO} the parameters
$\bH$ and $\bR_s$ are both unknown, so it is convenient to replace them by  $\tilde{\bR}_s=\bH \bR_s \bH^{\dag}$. Thus, the log-likelihood under $H_1$ can be written as
\begin{align}
\label{eq:log_likelihood_H1_1}
&L_1(\bR, \tilde{\bR}_s, \gamma;  \bZ) =
-NK \log \pi -NK_S \log \gamma \nonumber 
\\ 
\nonumber &- K_P\log\det(\tilde{\bR}_s + \bR)
-\tr\left[
\left( \tilde{\bR}_s+\bR\right)^{-1}
\bS_P
\right] 
\\ 
&- K_S\log\det(\bR)
-\tr\left[
\frac{1}{\gamma} \bR^{-1}
\bS_S
\right]
\end{align}
where we recall that
$\bS_P = 
\bZ_P \bZ_P^{\dag}$ and
$\bS_S =  \bZ_S \bZ_S^{\dag}$ (and the matrix $\bS_S$ is positive definite since $K_S \geq N$). Notice also that the rank of the matrix $\tilde{\bR}_s$ is less than or equal to $r$ 
(in fact, the rank of $\bH \bR_s^{1/2}$ is less than or equal to $r$). 
Now, consider the eigendecomposition of 
\be
\bR^{-1/2}(\tilde{\bR}_{s}+\bR)\bR^{-1/2}=\bU\bLambda\bU^\dag
\ee
where $\bLambda=\diag(\lambda_1,\ldots,\lambda_N) \in\R^{N\times N}$, $\lambda_1\geq\ldots\geq \lambda_N\geq 1$, is a diagonal matrix 
containing the eigenvalues\footnote{The $\lambda_i$s are greater than or equal to one since $\bR^{-1/2}(\tilde{\bR}_{s}+\bR)\bR^{-1/2}=\bI_N+\bR^{-1/2}\tilde{\bR}_{s}\bR^{-1/2}$.} of $\bR^{-1/2}(\tilde{\bR}_{s}+\bR)\bR^{-1/2}$ and
$\bU\in\C^{N\times N}$ is the unitary matrix of the corresponding eigenvectors. 
Define
$\bM=\bR^{1/2}\bU\in\C^{N\times N}$. Notice that the matrix $\bM$ can be any $N \times N$ non-singular matrix; in fact, from the eigendecomposition of the non-singular matrix $\bR^{1/2}$, namely 
$\bR^{1/2}=\bW_1 \bSigma \bW_1^{\dag}$, it follows that $\bM=\bW_1 \bSigma \bW_1^{\dag}\bU=\bW_1 \bSigma \bW_2^{\dag}$. In addition, we obtain that
\be
\bR=\bM\bM^\dag \quad \mbox{and} \quad \bR+\tilde{\bR}_s=\bM\bLambda\bM^\dag
\label{eq:R&RstildevsM&Lambda}
\ee
and \eqref{eq:log_likelihood_H1_1} can be recast as
\begin{align}
\nonumber
&L_1(\bR, \tilde{\bR}_s, \gamma; \bZ) =
-NK \log \pi -NK_S \log \gamma 
\\ 
\nonumber & -2K_P\log |\det \bM| -K_P\log\det \bLambda - 2K_S\log |\det \bM |
\\ 
&-
\tr\left[
\bLambda^{-1} \bM^{-1} \bS_P \bM^{-\dag}
\right] - \frac{1}{\gamma}\tr\left[
 \bM^{-1}
\bS_S \bM^{-\dag}
\right]
\label{eq:log_likelihood_H1_2}
\end{align}
where we have used the facts that 
$\det\bM^{\dag}=\left(\det(\bM)\right)^*$
and $\tr(\bA \bB)=\tr(\bB \bA)$, see for instance \cite{MatrixAnalysis}.
Moreover, denote by $\bGamma=\diag(\gamma_1,\ldots,\gamma_N)\in\R^{N\times N}$, $\gamma_1\geq\ldots\geq\gamma_N\geq 0$, the diagonal matrix 
containing the eigenvalues of $\bS_S^{-1/2}\bS_P\bS_S^{-1/2}$ and by $\bV\in\C^{N\times N}$ the unitary 
matrix of the corresponding eigenvectors.
Define $\bK=\bS_S^{1/2}\bV
\in\C^{N\times N}$; it turns out that 
$\bS_S=\bK\bK^\dag$ and $\bS_P=\bK\bGamma\bK^\dag$. Thus, we can rewrite
eq.  \eqref{eq:log_likelihood_H1_2} as
\begin{align}
\nonumber
&L_1(\bR, \tilde{\bR}_s, \gamma; \bZ)
=-NK \log \pi -NK_S \log \gamma 
\\ 
\nonumber &-2K_P\log |\det \bM| 
-K_P\log\det \bLambda -2K_S\log |\det \bM| 
\\ 
\nonumber &- \tr\left[
 \bLambda^{-1} \bM^{-1} \bK\bGamma\bK^\dag \bM^{-\dag}
\right] -\frac{1}{\gamma}\tr\left[
\bM^{-1}\bK\bK^\dag\bM^{-\dag}
\right] 
\\ 
\nonumber
&=-NK \log \pi -NK_S \log \gamma 
-2K\log |\det \bM|   
\\ 
& -K_P\log\det \bLambda - \tr\left[
  \bX \bGamma \bX^\dag
\right] 
-\frac{1}{\gamma}\tr\left[\bLambda^{1/2}
\bX  \bX^\dag \bLambda^{1/2}
\right]
\label{eq:log_likelihood_H1_3}
\end{align}
where $\bX=\bLambda^{-1/2}\bM^{-1}\bK\in\C^{N\times N}$ and we recall that $\bK$ and $\bGamma$ are known. Before going further, we also observe that we can maximize over $\bX$ for any given $\bLambda$ since $\bX$, given $\bLambda$, is completely specified by $\bM$. Let us proceed by replacing $\bX$ 
in \eqref{eq:log_likelihood_H1_3} with its
singular value decomposition given by $\bX=\bT\bD\bQ$, where $\bT,\bQ\in\C^{N\times N}$ are
unitary matrices and\footnote{Notice that the singular values of $\bX$ have been arranged in ascending 
order and not, as customary, in descending order \cite{MatrixAnalysis}.} $\bD=\diag(d_1,\ldots,d_N)\in\R^{N\times N}$, $0 <d_1\leq\ldots\leq d_N$. 
Since
\begin{align*}
\log |\det(\bM)| &=\log |\det(\bK)|
-\log |\det(\bX)| -\frac{1}{2}\log\det(\bLambda)
\\ &=
\log |\det(\bK)|
-\log\det(\bD) -\frac{1}{2}\log\det(\bLambda)
\end{align*}
we obtain 
\begin{align}
\nonumber
&L_1(\bR, \tilde{\bR}_s, \gamma; \bZ)
=-NK \log \pi -NK_S \log \gamma 
\\ 
\nonumber &- 2K \log |\det(\bK) |
+ 2K
\log\det(\bD)+ K_S \log\det(\bLambda)
\\ 
& - \tr\left[
 \bGamma \bQ^\dag \bD^2 \bQ 
\right] 
- \frac{1}{\gamma} \tr\left[
\bLambda \bT \bD^2 \bT^\dag
\right].
\label{eq:log_likelihood_H1_4}
\end{align}
Maximization with respect to $\bR$ and $\tilde{\bR}_s$ is tantamount to maximizing with respect to
$\bLambda, \bT, \bD, \bQ$. 
Therefore, we have that
\begin{align}
\nonumber
&L_1(\widehat{\bR}, \widehat{\tilde{\bR}}_s, \gamma; \bZ)
=-NK \log \pi -NK_S \log \gamma 
\\ \nonumber &- 2K\log |\det(\bK) |+ 
\dmax_{\bLambda, \bT, \bD, \bQ}
\Big\{ 2K\log\det(\bD) 
\\ 
& +K_S\log\det(\bLambda)
- \tr\left[
 \bGamma \bQ^\dag \bD^2 \bQ 
\right] 
- \frac{1}{\gamma} \tr\left[
\bLambda \bT \bD^2 \bT^\dag
\right]
\Big\}.
\label{eq:log_likelihood_H1_5}
\end{align}
Exploiting {\em Theorem 1} of \cite{mirsky1959trace}, it turns out that
\begin{align*}
\dmin_{\bQ}
\tr\left[
 \bGamma \bQ^\dag \bD^2 \bQ 
\right] 
&+ \frac{1}{\gamma} \dmin_{\bT} \tr\left[
\bLambda \bT \bD^2 \bT^\dag
\right]
\\ &=
\tr\left[
 \left(\bGamma + \frac{1}{\gamma} \bLambda \right) \bD^2 
\right]
\end{align*}
where the maximizers are $\widehat{\bT}=\bI_N e^{j\theta_T}$, $\forall\theta_T\in[0, 2\pi]$, 
$\widehat{\bQ}=\bI_N e^{j\theta_Q}$, $\forall\theta_Q\in[0, 2\pi]$, and hence
\begin{align}
\nonumber
&L_1(\widehat{\bR}, \widehat{\tilde{\bR}}_s, \gamma; \bZ)
=-NK \log \pi -NK_S \log \gamma 
\\  
\nonumber &- 2K \log |\det(\bK) |+ 
\dmax_{\bLambda, \bD}
\bigg\{ 2K\log\det(\bD) + K_S\log\det(\bLambda)
\\  \nonumber & -\tr\left[
 \left(\bGamma + \frac{1}{\gamma} \bLambda \right) \bD^2 
\right]
\bigg\}=-NK \log \pi -NK_S \log \gamma 
\\  \nonumber &- 2K \log |\det(\bK) | + \dmax_{\bLambda, \bD}
\left[\ K\sum_{i=1}^N\log d_i^2
+ K_S \sum_{i=1}^N\log \lambda_i \right.
\\  &\left. -\sum_{i=1}^N d_i^2 
\left( \gamma_i + \frac{1}{\gamma} \lambda_i \right)
\right].
\label{eq:log_likelihood_H1_5}
\end{align}
Now we compute the maximum with respect to
$d^2_i$, given $\lambda_i$, of the function
$$
g_i(d_i^2)=
K \log d_i^2
+ K_S \log \lambda_i 
- d_i^2 \left( \gamma_i + \frac{1}{\gamma} \lambda_i \right).
$$
First notice that 
$g_i$ tends to $-\infty$ as $d_i^2$ tends to zero or to $+\infty$; moreover,
its derivative  with respect to
$d^2_i$ is positive iff
$$
\frac{K}{d_i^2} - \left( \gamma_i + \frac{1}{\gamma} \lambda_i \right)  > 0
$$
implying that the maximizer (given $\lambda_i$) is
$d_i^2= \frac{\gamma K}{\gamma \gamma_i + \lambda_i}$.
It follows that eq. (\ref{eq:log_likelihood_H1_5}) 
yields
\begin{align}
\nonumber
&L_1(\widehat{\bR}, \widehat{\tilde{\bR}}_s, \gamma; \bZ)
=-NK \log \pi -NK_S \log \gamma 
- 2K \log |\det(\bK) |
\\ 
&+ \dmax_{\bLambda}
\left[\ K\sum_{i=1}^N\log \frac{\gamma K}{\gamma \gamma_i + \lambda_i} + K_S \sum_{i=1}^N\log \lambda_i 
-NK
\right].
\label{eq:log_likelihood_H1_6}
\end{align}
Now observe that even though the rank of $\bH$ is 
known (and equal to $r$) that of $\tilde{\bR}_s$ is unknown (less than or equal to $r$).
Thus, maximization over $\bLambda$ of eq. (\ref{eq:log_likelihood_H1_6}) is limited to $N \times N$ diagonal matrices 
with (at most) $r$ eigenvalues greater than one and the remaining equal to one. Hereafter we focus on the case $r \leq K_P \leq N$ although extension to $K_P<r$ is straighforward. 

\begin{theorem}
\label{TheoremSO_US_HE}
Let $r \leq K_P \leq N$. Eq. (\ref{eq:log_likelihood_H1_6}) can be rewritten as 
\begin{align}
\nonumber
&L_1(\widehat{\bR}, \widehat{\tilde{\bR}}_s, \gamma; \bZ)
=-NK \log \pi -NK_S \log \gamma -NK
\\ \nonumber &- 2K \log |\det(\bK)| +
\sum_{i=1}^{r-1} K \log \frac{\gamma K}{\gamma \gamma_i + \widehat{\lambda}_i(\gamma)}
\\ &+  \sum_{i=1}^{r-1} K_S \log \widehat{\lambda}_i (\gamma)
+ \sum_{i={r}}^{N} K \log \frac{\gamma K}{\gamma \gamma_{i} + 1}
\label{eq:log_likelihood_H1_theoremSO_US_HE_1}
\end{align}
with
$
\widehat{\lambda}_i=\max \left(\frac{K_S \gamma \gamma_i}{K_P},1 \right)$, $i=1,\ldots,r-1$,
if $\gamma < \frac{K_P}{K_S}\frac{1}{\gamma_{r}}$,
and as
\begin{align}
\nonumber
&L_1(\widehat{\bR}, \widehat{\tilde{\bR}}_s, \gamma; \bZ)
=-NK \log \pi -NK_S \log \gamma -NK
\\ \nonumber &- 2K \log |\det(\bK)| +
\sum_{i={1}}^{r} \left[ K \log \frac{K_P}{\gamma_i}+ K_S \log \frac{K_S \gamma \gamma_i}{K_P} \right]
\\ &+
 \sum_{i={r+1}}^{N}K \log \frac{\gamma K}{\gamma \gamma_{i} + 1}
\label{eq:log_likelihood_H1_theoremSO_US_HE_2}
\end{align}
otherwise.
\end{theorem}

\begin{proof}
See Appendix \ref{App:ProofTheoremSO_US_HE}.
\end{proof}

Obviously, if $r= N$ 
the compressed likelihood under $H_1$ is given from
(\ref{eq:log_likelihood_H1_theoremSO_US_HE_1}) by removing the last term and varying the summation indices from $1$ to $N$.
Notice also that 
$K_P \geq h$  is a necessary condition to  estimate 
the eigenvalues of a matrix $\tilde{\bR}_s$ with rank $h$.

Part of the above results, useful for subsequent derivations, is now restated in the form of 
theorem.
In particular, the following theorem summarizes the achieved findings in the case that the rank of $\tilde{\bR}_s$ is unknown.

\begin{theorem}
\label{eq:ML_estimate_two_cov}
Let $\bS_P \in \C^{N \times N}$ be a positive semidefinite matrix and
$\bS_S \in \C^{N \times N}$ a positive definite matrix. Then,
the function
\begin{align*}
h( \bR, \tilde{\bR}_s, \gamma) &=-N K_S \log \gamma
-K_P\log\det(\tilde{\bR}_s + \bR)
\\ &-K_S\log\det(\bR)
\\ &-\tr\left[
\left( \tilde{\bR}_s+\bR\right)^{-1}
\bS_P
\right] 
-\frac{1}{\gamma} \tr\left[
\bR^{-1}
\bS_S
\right]
\end{align*}
for any $\gamma >0$  admits its maximum over the set of positive definite matrices $\bR \in \C^{N \times N}$ and positive semidefinite matrices $\tilde{\bR}_s \in \C^{N \times N}$ with unknown rank at 
$$
\widehat{\bR}=\widehat{\bM} \widehat{\bM}^\dag \quad \mbox{and} \quad \widehat{\tilde{\bR}}_s=\widehat{\bM} \widehat{\bLambda} \widehat{\bM}^\dag
-\widehat{\bR}
$$
where
$\widehat{\bLambda}=\diag(\widehat{\lambda}_1,\ldots,\widehat{\lambda}_N) \in\R^{N\times N}$, 
with
$$
\widehat{\lambda}_i(\gamma)=\max\left(\frac{K_S}{K_P}\gamma \gamma_i ,1\right), \quad i=1,\ldots,N,
$$ 
and
$\gamma_1\geq\ldots\geq\gamma_N \geq 0$ the eigenvalues of $\bS_S^{-1/2} \bS_P \bS_S^{-1/2}$
(define also
the unitary matrix $\bV\in\C^{N\times N}$ of the corresponding eigenvectors
of $\bS_S^{-1/2} \bS_P \bS_S^{-1/2}$)
while
$$
\widehat{\bM}=
\bK  \widehat{\bD^2}^{-1/2} \widehat{\bLambda}^{-1/2}
$$
with 
$\widehat{\bD^2}=\diag\left(\widehat{d^2}_1,\ldots,\widehat{d^2}_N\right)\in\R^{N\times N}$,
$$
\widehat{d_i}^2(\gamma)= \frac{\gamma K}{\gamma \gamma_i + \widehat{\lambda}_i(\gamma)}, \quad i=1,\ldots,N,
$$
and
$\bK=\bS_S^{1/2}\bV
\in\C^{N\times N}$.
The maximum is given by
\begin{align}
\nonumber
h(\widehat{\bR}, \widehat{\tilde{\bR}}_s, \gamma)
&=-NK_S \log \gamma 
- 2K \log |\det(\bK) | -NK
\\ &+ 
 \ K\sum_{i=1}^N\log \frac{\gamma K}{\gamma \gamma_i + \widehat{\lambda}_i(\gamma)}
+ K_S \sum_{i=1}^N\log \widehat{\lambda}_i (\gamma)
\label{eq:maximum_h_given_gamma}
\end{align}
with $K=K_P+K_S$.
\medskip
\end{theorem}

The GLRT for homogeneous environment and unknown subspace $\langle \bH \rangle$, referred to in the following as 
second-order unknown subspace in homogeneous environment 
(SO-US-HE) detector,
is  
\be
L_1(\widehat{\bR}, \widehat{\tilde{\bR}}_s, 1; \bZ)
-L_0(\widehat{\bR},  1; \bZ)
\test \eta
\label{eq:1S-SO-GLRT-UK-HE}
\ee
with $L_0(\widehat{\bR},  1;  \bZ)$ given by the logarithm of eq. (\ref{eq:PDF_H0_FO_knownH}) (with $\gamma=1$).
\medskip

\subsection{Unknown subspace $\langle \bH \rangle$ of known dimension, unknown  $\gamma$}
To derive the GLRT for partially-homogeneous environment,
we have to maximize the partially-compressed likelihood  (under $H_1$), given by Theorem \ref{TheoremSO_US_HE},
also with respect to $\gamma$. Such maximization is summarized  
by the following theorem.
\bigskip

\begin{theorem}
\label{Theorem4}
Let $r < K_P \leq N$.
The maximum with respect to $\gamma$ of  the
partially-compressed likelihood, given by Theorem \ref{TheoremSO_US_HE},
is attained at the unique $\gamma \geq \frac{K_P}{K_S} \frac{1}{\gamma_r}$, say $\widehat{\gamma}$, solving
the equation
$$
\sum_{i=r+1}^{K_P} \frac{K}{\gamma \gamma_i+1}= (K_P-r)K_S- (N-K_P)K_P
$$
provided that $(K_P-r)K_S > (N-K_P)K_P$.
The compressed likelihood is obtained plugging $\widehat{\gamma}$ into eq. (\ref{eq:log_likelihood_H1_theoremSO_US_HE_2}).
\end{theorem}

\begin{proof}
See Appendix \ref{App:ProofTheorem4}.
\end{proof}

Finally, the GLRT, referred to in the following as
second-order unknown subspace in partially-homogeneous environment
(SO-US-PHE) detector, is given by
\be
L_1(\widehat{\bR}, \widehat{\tilde{\bR}}_s, \widehat{\gamma}; \bZ)-L_0(\widehat{\bR},  \widehat{\gamma}; \bZ)
\test \eta
\label{eq:1S-SO-GLRT-UH-PHE}
\ee
with
$L_0(\widehat{\bR},  \widehat{\gamma}; \bZ)$ given by the logarithm of the maximum of eq. (\ref{eq:PDF_H0_FO_knownH}) with respect to $\gamma$
obtained by using 
Theorem \ref{theorem_minimization_over_gamma_FH_PHE}.

\subsection{Known subspace $\langle \bH \rangle$, known  $\gamma$}

As a first step towards the computation of the GLRT, we extend \cite{BCCRV} where the 
case $\bH \in \C^{N}$ (rank-one signal) and $\gamma=1$ (homogeneous environment) is addressed. To this end, 
we
denote by $\bH_{\perp} \in \C^{N \times (N-r)}$ a slice of unitary matrix spanning the orthogonal complement of $\bH$.
It follows that the matrix $\bV=[ \bH \ \bH_{\perp}]\in\C^{N\times N}$ is a unitary one. 
Then, we rewrite the likelihoods under $H_1$ and $H_0$ as
\begin{align*}
\ell_1(\bR, \bR_s, \bH, \gamma; \bZ) &=
\frac{\etr{-
\left(\tilde{\bR} +\bE \bR_s \bE^\dag\right)^{-1}   
\tilde{\bZ}_P \tilde{\bZ}_P^{\dag}
}}
{\pi^{NK} \gamma^{NK_S} \det^{K_P}(\tilde{\bR} +\bE \bR_s \bE^\dag)}
\\ &\times \frac{\etr{-\frac{1}{\gamma} \tilde{\bR}^{-1} 
\tilde{\bZ}_S \tilde{\bZ}_S^{\dag}
}}
{\det^{K_S}(\tilde{\bR})}
\end{align*}
and
\[
\ell_0(\bR, \gamma; \bZ) =
\frac{\etr{-\left[ \tilde{\bR}^{-1} 
\tilde{\bZ}_P \tilde{\bZ}_P^{\dag}+
\frac{1}{\gamma}
\tilde{\bR}^{-1} 
\tilde{\bZ}_S \tilde{\bZ}_S^{\dag}
\right]}}
{\pi^{NK} \gamma^{NK_S} \det^{K}(\tilde{\bR})},
\]
respectively, where $\tilde{\bZ}_P= \bV^{\dag} \bZ_P=\left[\tilde{\bZ}_{P,1}^T \ \tilde{\bZ}_{P,2}^T \right]^T$, with
$\tilde{\bZ}_{P,1} \in \C^{r \times K_P}$
and $\tilde{\bZ}_{P,2} \in \C^{(N-r) \times K_P}$,
$\tilde{\bZ}_S= \bV^{\dag} \bZ_S=\left[\tilde{\bZ}_{S,1}^T \ \tilde{\bZ}_{S,2}^T \right]^T$, with
$\tilde{\bZ}_{S,1} \in \C^{r \times K_S}$
and $\tilde{\bZ}_{S,2} \in \C^{(N-r) \times K_S}$,
\begin{align*}
\tilde{\bR} &= \bV^{\dag} \bR \bV=
\left[
\begin{array}{cc}
\bH^{\dag} \bR \bH & \bH^{\dag} \bR \bH_{\perp} \\
\bH^\dag_{\perp}  \bR \bH & \bH^\dag_{\perp}  \bR \bH_{\perp}
\end{array}
\right]
\end{align*}
while $\bE=[\bI_r \ \bzero_{r,N-r}]^T$ and, obviously, also
\begin{align*}
\bV^{\dag} \left( \bH \bR_s \bH^{\dag}+\bR\right) \bV &=
\left[
\begin{array}{cc}
\bR_s + \bH^{\dag} \bR \bH & \bH^{\dag} \bR \bH_{\perp} \\
\bH^\dag_{\perp}  \bR \bH & \bH^\dag_{\perp}  \bR \bH_{\perp}
\end{array}
\right]
\\ &=
\tilde{\bR} +\bE \bR_s \bE^\dag.
\end{align*}
We now observe that \cite{Handbook_of_matrices,MatrixAnalysis}
\begin{eqnarray*}
\tilde{\bR}^{-1} &=& 
\left[
\begin{array}{cc}
\tilde{\bR}_{11}  & \tilde{\bR}_{12} \\
\tilde{\bR}_{21} & \tilde{\bR}_{22}
\end{array}
\right]^{-1}
\\ &=&
\left[
\begin{array}{cc}
\tilde{\bR}_{1.2}^{-1}  & -\tilde{\bR}_{1.2}^{-1}\bbeta^{\dag} \\
- \bbeta\tilde{\bR}_{1.2}^{-1} & \tilde{\bR}_{22}^{-1} +
\bbeta
\tilde{\bR}_{1.2}^{-1}
\bbeta^{\dag}
\end{array}
\right]
\\ &=&
\bB^\dag
\tilde{\bR}_{1.2}^{-1}
\bB+
\left[
\begin{array}{cc}
\bzero_{r,r}  & \bzero_{r,N-r} \\
\bzero_{N-r,r} & \tilde{\bR}_{22}^{-1}
\end{array}
\right]
\end{eqnarray*}
with $\bB=[\bI_r \ -\bbeta^\dag]\in\C^{r \times N}$, 
$\tilde{\bR}_{1.2}=\tilde{\bR}_{11}-\tilde{\bR}_{12}\tilde{\bR}_{22}^{-1}\tilde{\bR}_{21} \in \C^{r \times r}$, 
and $\bbeta=
\tilde{\bR}_{22}^{-1}\tilde{\bR}_{21} \in \C^{(N-r) \times r}$.
Similarly, we have that
\begin{align*}
&\left( \tilde{\bR} +\bE \bR_s \bE^\dag \right)^{-1} =
\left[
\begin{array}{cc}
\bR_s+\tilde{\bR}_{11}  & \tilde{\bR}_{12} \\
\tilde{\bR}_{21} & \tilde{\bR}_{22}
\end{array}
\right]^{-1}
\\ &=
\left[
\begin{array}{cc}
\left( \bR_s+\tilde{\bR}_{1.2}\right)^{-1}  & -\left( \bR_s+\tilde{\bR}_{1.2}\right)^{-1}\bbeta^{\dag} \\
- \bbeta\left( \bR_s+\tilde{\bR}_{1.2}\right)^{-1} & \tilde{\bR}_{22}^{-1} +
\bbeta
\left( \bR_s+\tilde{\bR}_{1.2}\right)^{-1}
\bbeta^{\dag}
\end{array}
\right]
\\ &=
\bB^\dag
\left( \bR_s+\tilde{\bR}_{1.2}\right)^{-1}
\bB
+
\left[
\begin{array}{cc}
\bzero_{r,r}  & \bzero_{r,N-r} \\
\bzero_{N-r,r} & \tilde{\bR}_{22}^{-1}
\end{array}
\right].
\end{align*}
Moreover, we have that \cite{Handbook_of_matrices,MatrixAnalysis}
$$
\det \tilde{\bR}= \det \tilde{\bR}_{22} \cdot \det \tilde{\bR}_{1.2}
$$
and
$$
\det \left( \tilde{\bR} +\bE \bR_s \bE^\dag \right) = \det \tilde{\bR}_{22} \cdot \det \left( \bR_s + \tilde{\bR}_{1.2} \right).
$$
It follows that 
\begin{align*}
&\ell_1(\bR, \bR_s, \bH, \gamma; \bZ) =
\frac{1}{\pi^{NK}} \frac{1}{\gamma^{NK_S}} 
\frac{1}{\det^{K}(\tilde{\bR}_{22})}
\\ &\times
\frac{\etr{-\left[
\tilde{\bR}_{22}^{-1}   
\tilde{\bZ}_{P,2} \tilde{\bZ}_{P,2}^{\dag} +
\frac{1}{\gamma} \tilde{\bR}_{22}^{-1} 
\tilde{\bZ}_{S,2} \tilde{\bZ}_{S,2}^{\dag} 
\right]}}
{\det^{K_P}(\bR_s + \tilde{\bR}_{1.2}) \det^{K_S}(\tilde{\bR}_{1.2})}
\\ &\times
\etr{-
\left( \bR_s+\tilde{\bR}_{1.2}\right)^{-1}
\bB  
\tilde{\bZ}_P \tilde{\bZ}_P^{\dag}
\bB^\dag
}
\\ &\times
\etr{- \frac{1}{\gamma}
\tilde{\bR}_{1.2}^{-1}
\bB   
\tilde{\bZ}_S \tilde{\bZ}_S^{\dag}
\bB^\dag
}.
\end{align*}
\medskip

Subsequent developments rely on the fact that 
we can estimate the parameters $\tilde{\bR}_{1.2}, \tilde{\bR}_{22}, \bbeta, \bR_s$ 
in place of $\bR, \bR_s$.
To this end, first observe that
the ML estimate of $\tilde{\bR}_{22}$, given $\gamma$, can be expressed as
$$
\widehat{\tilde{\bR}}_{22}=
\frac{1}{K} \left( \tilde{\bZ}_{P,2} \tilde{\bZ}_{P,2}^{\dag} +
\frac{1}{\gamma} \tilde{\bZ}_{S,2} \tilde{\bZ}_{S,2}^{\dag} \right)
$$
and the corresponding partially-compressed likelihood
becomes 
\begin{align}
\nonumber
&\dmax_{\tilde{\bR}_{22}} \ell_1(\bR, \bR_s, \bH, \gamma; \bZ) =
\frac{1}{\pi^{NK}} \frac{1}{\gamma^{NK_S}}
\frac{\left( {K}/{e} \right)^{(N-r)K}}
{\det^{K_P}(\bR_s + \tilde{\bR}_{1.2})} 
\\ \nonumber &\times
\frac{\etr{- 
\left( \bR_s+\tilde{\bR}_{1.2}\right)^{-1}
\bB   
\tilde{\bZ}_P \tilde{\bZ}_P^{\dag}
\bB^\dag}}
{\det^{K_S}(\tilde{\bR}_{1.2}) \det^{K}\left( \tilde{\bZ}_{P,2} \tilde{\bZ}_{P,2}^{\dag} +
\frac{1}{\gamma} \tilde{\bZ}_{S,2} \tilde{\bZ}_{S,2}^{\dag} \right)}
\\ &\times
\etr{- \frac{1}{\gamma}
\tilde{\bR}_{1.2}^{-1}
\bB   
\tilde{\bZ}_S \tilde{\bZ}_S^{\dag}
\bB^\dag }.
\label{eq:likelihood_compressed_with_respect_to_R22}
\end{align}
\medskip

Estimation of the remaining parameters
cannot be conducted in closed form at the best of authors' knowledge. For this reason we implement an alternating maximization \cite{Stoica_alternating} which
estimates a subset of the unknown parameters assuming that the remaining parameters are known and vice versa.
In particular, we exploit the following results.

\subsubsection{Estimate of $\tilde{\bR}_{1.2}$ and $\bR_s$, given $\bbeta$}

We write the logarithm of the partially-compressed likelihood 
\eqref{eq:likelihood_compressed_with_respect_to_R22} as follows

\begin{align}
\nonumber
&\dmax_{\tilde{\bR}_{22}}
L_1(\bR, \bR_s, \bH, \gamma; \bZ) =
C -NK_S \log \gamma 
\\ \nonumber &- K_P\log\det(\bR_s + \tilde{\bR}_{1.2})-K_S\log\det(\tilde{\bR}_{1.2})
\\ \nonumber &-
K \log \det\left( \tilde{\bZ}_{P,2} \tilde{\bZ}_{P,2}^{\dag} +
\frac{1}{\gamma} \tilde{\bZ}_{S,2} \tilde{\bZ}_{S,2}^{\dag} \right)
\\ 
&-
\tr\left[
\left( \bR_s+\tilde{\bR}_{1.2}\right)^{-1}
\tilde{\bS}_P
\right] 
- \frac{1}{\gamma} \tr\left[
\tilde{\bR}_{1.2}^{-1}
\tilde{\bS}_S
\right]
\label{eq:log_likelihood_compressed_with_respect_to_R22_log}
\end{align}
where 
$\tilde{\bS}_P = \bB \tilde{\bZ}_P \tilde{\bZ}_P^{\dag}\bB^\dag$
and $\tilde{\bS}_S = \bB \tilde{\bZ}_S \tilde{\bZ}_S^{\dag} \bB^\dag$,
while $C=-NK \log \pi +(N-r)K \left(\log K - \log e \right)$ gathers the terms that are irrelevant to the maximization. 
It can be shown that $K_S \geq N$ makes $\tilde{\bS}_S$ a non-singular matrix.
Exploiting Theorem \ref{eq:ML_estimate_two_cov}, with $\bR_s$ and $\tilde{\bR}_{1.2}$ in place of $\tilde{\bR}_s$, and $\bR$,
respectively, (notice also that the matrices are $r \times r$ in place of $N \times N$)
we obtain that
\begin{align}
\nonumber
&\dmax_{\tilde{\bR}_{22},\tilde{\bR}_{1.2}, \bR_s} 
L_1(\bR, \bR_s, \bH, \gamma; \bZ)
= C
\\ \nonumber &-
K \log \det\left( \tilde{\bZ}_{P,2} \tilde{\bZ}_{P,2}^{\dag} +
\frac{1}{\gamma} \tilde{\bZ}_{S,2} \tilde{\bZ}_{S,2}^{\dag} \right)
\\ \nonumber 
&- NK_S \log \gamma 
- 2K \log |\det(\bK)| -rK
\\ &+
 \ K\sum_{i=1}^r\log \frac{\gamma K}{\gamma \gamma_i + \widehat{\lambda}_i(\gamma)}
+ K_S \sum_{i=1}^r\log \widehat{\lambda}_i (\gamma)
\label{eq:log_likelihood_compressed_with_respect_to_R22_R1dot2_Rs_log}
\end{align}
with $\gamma_1\geq\ldots\geq\gamma_r \geq 0$ the eigenvalues of $\tilde{\bS}_S^{-1/2} \tilde{\bS}_P \tilde{\bS}_S^{-1/2}
\in\C^{r\times r}$,
$\bV\in\C^{r\times r}$ the unitary 
matrix of the corresponding eigenvectors
of $\tilde{\bS}_S^{-1/2} \tilde{\bS}_P \tilde{\bS}_S^{-1/2}$, 
$\bK=\tilde{\bS}_S^{1/2}\bV
\in\C^{r\times r}$, and
$$
\widehat{\lambda}_i(\gamma)=\max\left(\frac{K_S}{K_P}\gamma \gamma_i ,1\right), \quad i=1,\ldots,r.
$$ 
Notice that application of the theorem returns the following estimates of $\tilde{\bR}_{1.2}$ and $\bR_s$
$$
\widehat{\tilde{\bR}}_{1.2}=\widehat{\bM} \widehat{\bM}^\dag \quad \mbox{and} \quad \widehat{\bR}_s=\widehat{\bM} \widehat{\bLambda} \widehat{\bM}^\dag
-\widehat{\tilde{\bR}}_{1.2}
$$
where
$\widehat{\bLambda}=\diag(\widehat{\lambda}_1,\ldots,\widehat{\lambda}_r) \in\R^{r\times r}$
while $\widehat{\bM}= \bK  \widehat{\bD^2}^{-1/2} \widehat{\bLambda}^{-1/2} $ with 
$\widehat{\bD^2}=\diag\left(\widehat{d^2}_1,\ldots,\widehat{d^2}_r\right)\in\R^{r\times r}$,
$$
\widehat{d_i}^2(\gamma)= \frac{\gamma K}{\gamma \gamma_i + \widehat{\lambda}_i(\gamma)}, \quad i=1,\ldots,r.
$$
\medskip

\subsubsection{Estimate of $\bbeta$ given $\tilde{\bR}_{1.2}$ and $ \bR_s$}

First we observe that
\begin{align*}
&\bB   
\tilde{\bA}
\bB^\dag =
\tilde{\bA}_{11}-\bbeta^{\dag}  \tilde{\bA}_{21}
- \tilde{\bA}_{12} \bbeta 
+ \bbeta^{\dag}\tilde{\bA}_{22}
\bbeta
\end{align*}
where 
$$
\tilde{\bA}=\tilde{\bZ}_P \tilde{\bZ}_P^{\dag}=
\left[
\begin{array}{cc}
\tilde{\bA}_{11} & \tilde{\bA}_{12} \\
\tilde{\bA}_{21} & \tilde{\bA}_{22}
\end{array}
\right]
$$
with $\tilde{\bA}_{11} \in \C^{r \times r}$,
$\tilde{\bA}_{12} \in \C^{r \times (N-r)}$,
$\tilde{\bA}_{22} \in \C^{(N-r) \times (N-r)}$,
$\tilde{\bA}_{21} \in \C^{(N-r) \times r}$,
and 
\begin{align*}
&\bB   
\tilde{\bB}
\bB^\dag =
\tilde{\bB}_{11}-\bbeta^{\dag}  \tilde{\bB}_{21}
- \tilde{\bB}_{12} \bbeta 
+ \bbeta^{\dag}\tilde{\bB}_{22}
\bbeta
\end{align*}
where 
$$
\tilde{\bB}=\tilde{\bZ}_S \tilde{\bZ}_S^{\dag}=
\left[
\begin{array}{cc}
\tilde{\bB}_{11} & \tilde{\bB}_{12} \\
\tilde{\bB}_{21} & \tilde{\bB}_{22}
\end{array}
\right]
$$
with
$\tilde{\bB}_{11} \in \C^{r \times r}$,
$\tilde{\bB}_{12} \in \C^{r \times (N-r)}$,
$\tilde{\bB}_{22} \in \C^{(N-r) \times (N-r)}$,
$\tilde{\bB}_{21} \in \C^{(N-r) \times r}$.
Thus, maximization of the right-hand side of eq. (\ref{eq:likelihood_compressed_with_respect_to_R22})
with respect to $\bbeta$ is tantamount to minimizing the following function
\begin{eqnarray*}
g(\bbeta) &=&
\tr \ \left[ \left( \bR_s+\tilde{\bR}_{1.2}\right)^{-1}
\right.
\\ &\times& \left.
\left( -\bbeta^{\dag}  \tilde{\bA}_{21}
- \tilde{\bA}_{12} \bbeta 
+ \bbeta^{\dag}\tilde{\bA}_{22}
\bbeta \right) \right.
\\ &+& \left. \frac{1}{\gamma} \tilde{\bR}_{1.2}^{-1}
\left( -\bbeta^{\dag}  \tilde{\bB}_{21}
- \tilde{\bB}_{12} \bbeta 
+ \bbeta^{\dag}\tilde{\bB}_{22}
\bbeta \right)
\right].
\end{eqnarray*}
On the other hand, setting to zero the derivative of $g$
with respect to $\bbeta$, we have that
\begin{align*}
\frac{\partial }{\partial \bbeta}g(\bbeta) &=
\left( \tilde{\bA}_{22}^{T}
\bbeta^{*} - \tilde{\bA}_{12}^{T} \right) \left(  \bR_s + \tilde{\bR}_{1.2} \right)^{-T} 
\\ &+
\frac{1}{\gamma} \left( \tilde{\bB}_{22}^{T}
\bbeta^{*} - \tilde{\bB}_{12}^{T} \right)  \tilde{\bR}_{1.2}^{-T}= \bzero_{N-r,r}.
\end{align*}
The above equation can be rewritten as
\begin{align*}
&\tilde{\bA}_{22}^{T}
\bbeta^{*} \left( \tilde{\bR}_{1.2} + \bR_s \right)^{-T}
+ \frac{1}{\gamma}
\tilde{\bB}_{22}^{T}
\bbeta^* \tilde{\bR}_{1.2}^{-T}
=
\tilde{\bA}_{12}^{T} 
\\ &\times \left( \tilde{\bR}_{1.2} + \bR_s \right)^{-T}
+ \frac{1}{\gamma}
\tilde{\bB}_{12}^{T}
\tilde{\bR}_{1.2}^{-T}
\end{align*}
and, exploiting the identity  
7.2 (7) in \cite{Handbook_of_matrices}
$$
\makebox{vec} \left( \bA \bX \bB \right)= \left(\bB^T \otimes \bA \right) \makebox{vec}\bX,
$$
also as
$$
\left[ \left( \tilde{\bR}_{1.2} + \bR_s \right)^{-1} \otimes \tilde{\bA}_{22}^{T}
+ \frac{1}{\gamma} \tilde{\bR}_{1.2}^{-1} \otimes 
\tilde{\bB}_{22}^{T}
\right] \makebox{vec}\bbeta^*=
\makebox{vec} \bC
$$
with
$$
\bC=
\tilde{\bA}_{12}^{T} \left( \tilde{\bR}_{1.2} + \bR_s \right)^{-T}
+ \frac{1}{\gamma} 
\tilde{\bB}_{12}^{T}
\tilde{\bR}_{1.2}^{-T}.
$$
Thus, 
its solution is given by 
\begin{align}
\nonumber
\makebox{vec} \bbeta^* &=
\left[ \left( \tilde{\bR}_{1.2} + \bR_s \right)^{-1} \otimes \tilde{\bA}_{22}^{T}
+ \frac{1}{\gamma} \tilde{\bR}_{1.2}^{-1} \otimes 
\tilde{\bB}_{22}^{T}
\right]^{-1} 
\\ &\times \makebox{vec}\bC.
\label{eq:beta_solution}
\end{align}

Now we observe that the matrix $\gamma \bR$ can be estimated as the sample covariance matrix of the secondary data only and used to construct an estimate of $\bbeta=
\tilde{\bR}_{22}^{-1}\tilde{\bR}_{21}$; denoting by $\widehat{\bbeta}^{(0)}$ this starting value we can exploit previous results to obtain after $n$ iteration of alternating mazimization $\widehat{\tilde{\bR}}_{1.2}^{(n)}$, $\widehat{\bR}_s^{(n)}$, and $\widehat{\bbeta}^{(n)}$ that together with $\widehat{\tilde{\bR}}_{22}$ allow to compute
$\widehat{\bR}_s$ and $\widehat{\bR}$.

Finally, the GLRT, referred to in the following as
second-order known subspace in homogeneous environment
(SO-KS-HE) detector, is given by
\be
L_1(\widehat{\bR}, \widehat{\bR}_s, \bH, 1; \bZ)-L_0( \widehat{\bR},  1; \bZ)
\test \eta
\label{eq:1S-SO-GLRT-KH-HE}
\ee
with
$L_0(\widehat{\bR},  1; \bZ)$ given by the logarithm of eq. (\ref{eq:PDF_H0_FO_knownH}).

\subsection{Known subspace $\langle \bH \rangle$, unknown  $\gamma$}

Derivation of the GLRT for partially-homogeneous environment 
is still based on alternating maximization; this time
we estimate $\tilde{\bR}_{1.2}$, $\bR_s$, and $\gamma$, given $\bbeta$,
and again we estimate $\bbeta$, given $\tilde{\bR}_{1.2}$, $ \bR_s$, and $\gamma$, using eq. (\ref{eq:beta_solution}).
\smallskip

Estimating $\tilde{\bR}_{1.2}$, $\bR_s$, and $\gamma$, given $\bbeta$,
requires maximizing eq. (\ref{eq:log_likelihood_compressed_with_respect_to_R22_R1dot2_Rs_log}) with respect to $\gamma$. To this end, 
we define
$\tilde{\bS}_{P,2}=\tilde{\bZ}_{P,2} \tilde{\bZ}_{P,2}^{\dag}
\in\C^{(N-r)\times (N-r)}$,
 $\tilde{\bS}_{S,2}=
\tilde{\bZ}_{S,2} \tilde{\bZ}_{S,2}^{\dag}
\in\C^{(N-r)\times (N-r)}$, 
and
$\tilde{\bS}_{S,2}^{-1/2} \tilde{\bS}_{P,2} \tilde{\bS}_{S,2}^{-1/2}
\in\C^{(N-r)\times (N-r)}$.
It follows that
\begin{align}
\nonumber
&\dmax_{\tilde{\bR}_{22},\tilde{\bR}_{1.2}, \bR_s} 
L_1(\bR, \bR_s, \bH, \gamma; \bZ)
= C- K\log \det \left(\tilde{\bS}_{S,2} \right)
\\ \nonumber &- 2K \log |\det(\bK)| -rK
\\ \nonumber 
&- 
K \log \det\left( \tilde{\bS}_{S,2}^{-1/2} \tilde{\bS}_{P,2} \tilde{\bS}_{S,2}^{-1/2} +
\frac{1}{\gamma} \bI_{N-r} \right)
-NK_S \log \gamma 
\\ &+ 
 \ K\sum_{i=1}^r\log \frac{\gamma K}{\gamma \gamma_i + \widehat{\lambda}_i(\gamma)}
+ K_S \sum_{i=1}^r\log \widehat{\lambda}_i (\gamma).
\label{eq:log_likelihood_compressed_with_respect_to_R22_R1dot2_Rs_log1}
\end{align}
Thus, denoting by 
$\delta_1\geq\ldots\geq\delta_{N-r}\geq 0$ the eigenvalues of $\tilde{\bS}_{S,2}^{-1/2} \tilde{\bS}_{P,2} \tilde{\bS}_{S,2}^{-1/2}$, we also have that
\begin{align}
\nonumber
&\dmax_{\tilde{\bR}_{22},\tilde{\bR}_{1.2}, \bR_s} L_1( \bR, \bR_s, \bH, \gamma; \bZ)
= C- K\log \det \left(\tilde{\bS}_{S,2}\right)
\\ \nonumber &- 2K \log |\det(\bK)| -rK
\\ \nonumber 
&- 
K \sum_{i=1}^{N-r} \log \left( \frac{1}{\gamma} + \delta_i \right)
-NK_S \log \gamma 
\\ &+ 
 \ K\sum_{i=1}^r\log \frac{\gamma K}{\gamma \gamma_i + \widehat{\lambda}_i(\gamma)}
+ K_S \sum_{i=1}^r\log \widehat{\lambda}_i (\gamma).
\label{eq:log_likelihood_compressed_with_respect_to_R22_R1dot2_Rs_log2}
\end{align}
To maximize the partially-compressed likelihood of 
eq. (\ref{eq:log_likelihood_compressed_with_respect_to_R22_R1dot2_Rs_log2})
with respect to $\gamma$, we can use the following theorem.  For simplicity we assume $K_P \geq r$.
\bigskip

\begin{theorem}
\label{Theorem5}
The global maximum of the function
\begin{align*}
f(\gamma) &=
- 
K \sum_{i=1}^{N-r} \log \left( \frac{1}{\gamma} + \delta_i \right)
-NK_S \log \gamma 
\\ &+ 
 K\sum_{i=1}^r\log \frac{\gamma K}{\gamma \gamma_i + \widehat{\lambda}_i(\gamma)}
+ K_S \sum_{i=1}^r\log \widehat{\lambda}_i (\gamma),
\end{align*}
$K=K_P+K_S$, $K_P \geq r$, is attained at the unique solution over $\gamma \in (0, +\infty)$ of the
equation
\begin{align*}
&\sum_{i=1}^{\min(K_P,N-r)}  
\frac{1}{1+\gamma \delta_i}
+ \left( N-r-\min(K_P,N-r)\right) 
\\ &-\frac{(N-r) K_S}{K} =0,
\end{align*}
say $\gamma^*$, 
if $\gamma^* \geq \frac{K_P}{K_S} \frac{1}{\gamma_r}$.  Otherwise,
the stationary points of $f$ and, hence, its global maximum  belong to the interval $(\gamma^*, \frac{K_P}{K_S} \frac{1}{\gamma_r})$.
\end{theorem}

\begin{proof}
See Appendix \ref{App:ProofTheorem5}.
\end{proof}

Once we have determined $\widehat{\gamma}$, we can compute the GLRT, referred to in the following as 
second-order known subspace in partially-homogeneous environment
(SO-KS-PHE) detector, that can be written as
\be
L_1(\widehat{\bR}, \widehat{\bR}_s, \bH, \widehat{\gamma}; \bZ)-L_0(\widehat{\bR},  \widehat{\gamma}; \bZ)
\test \eta
\label{eq:1S-SO-GLRT-KH-PHE}
\ee
with
$L_0(\widehat{\bR},  \widehat{\gamma}; \bZ)$  
given by the logarithm of the maximum of eq. (\ref{eq:PDF_H0_FO_knownH}) with respect to $\gamma$
obtained by using 
Theorem \ref{theorem_minimization_over_gamma_FH_PHE}.

\section{Conclusions}
The original adaptive detectors  of \cite{Kelly86, CLR1995, Scharf-McWhorter}, as refined in \cite{Kraut-Scharf1999, Kraut-Scharf-McWhorter} and extended in \cite{CDMR2001,BDMGR}, have been generalized by considering three new subspace signal models.  In the first, the subspace visited by a sequence of symbol transmissions or reflections is assumed to be known only by its dimension; in the previous work by \cite{CDMR2001,BDMGR} the subspace was known. In the second  extension, a known subspace is visited by a sequence of symbol transmissions which are constrained  by a Gaussian prior distribution; the result is a second-order adaptive subspace detector. In the third extension, the subspace is known only by its dimension; this extension requires a two-channel extension of standard factor analysis. These extensions, coupled with the results of \cite{CDMR2001,BDMGR}, comprise a unified theory of adaptive subspace detection.

In a companion paper, detector performance will be compared against  subspace detectors that use an ad-hoc estimate of unknown noise covariance in known subspace detectors.

\appendices


\section{Proof of Theorem \ref{TheoremSO_US_HE}}
\label{App:ProofTheoremSO_US_HE}
First notice that the function
$$
h_i(\lambda_i)=
K \log \frac{\gamma K}{\gamma \gamma_i + \lambda_i}
+ K_S \log \lambda_i
$$
tends to $-\infty$ as $\lambda_i$ tends to $+\infty$; 
in addition, its derivative is positive iff
$$
-\frac{K}{\gamma \gamma_i + \lambda_i}
+\frac{K_S}{\lambda_i} >0
$$
or equivalently iff
$K_P \lambda_i < K_S \gamma \gamma_i$.
Thus, 
it follows that
$h_i$  has the global maximum over $[1,+\infty)$ at
$$
\tilde{\lambda}_i=\max \left(\frac{K_S \gamma \gamma_i}{K_P},1 \right).
$$
In particular, $\gamma_i=0$ implies $\tilde{\lambda}_i=0$.
It turns out that
\begin{itemize}
\item
if $\gamma < \frac{K_P}{K_S}\frac{1}{\gamma_1}$, all maximizers are equal to one and
$$
\max_{\lambda_i} h_i(\lambda_i)=
K \log \frac{\gamma K}{\gamma \gamma_i + 1}, \quad i=1, \ldots, N.
$$
Thus, the compressed likelihood under $H_1$ becomes
\begin{align}
\nonumber
&L_1(\widehat{\bR}, \widehat{\tilde{\bR}}_s, \gamma; \bZ)
=-NK \log \pi -NK_S \log \gamma 
\\ &- 2K \log |\det(\bK)| 
+ K\sum_{i={1}}^{N}\log \frac{\gamma K}{\gamma \gamma_i + 1}-NK.
\label{eq:log_likelihood_H1_7_known_rank-1}
\end{align}
\item
if $\frac{K_P}{K_S}\frac{1}{\gamma_{i-1}} \leq \gamma < \frac{K_P}{K_S}\frac{1}{\gamma_i}$, $i=2, \ldots, r$, 
the maximizers of $h_j$ are
$$
\tilde{\lambda}_j(\gamma)= 
\frac{K_S \gamma \gamma_j}{K_P}>1, \quad j=1, \ldots, i-1,
$$
and the remaining maximizers are $\tilde{\lambda}_j(\gamma)=1$, $j=i, \ldots, N$.
Thus, the
compressed likelihood under $H_1$ becomes
\begin{align}
\nonumber
&L_1(\widehat{\bR}, \widehat{\tilde{\bR}}_s, \gamma; \bZ)
=-NK \log \pi -NK_S \log \gamma 
\\ \nonumber &- 2K \log |\det(\bK)| +
\sum_{j={1}}^{i-1} \left[ K \log \frac{K_P}{\gamma_j}+ K_S \log \frac{K_S \gamma \gamma_j}{K_P} \right]
\\ &+
 \sum_{j={i}}^{N}K \log \frac{\gamma K}{\gamma \gamma_{j} + 1}
-NK.
\label{eq:log_likelihood_H1_7_known_rank-3}
\end{align}
\item
if $K_P=r$ and $\gamma \geq \frac{K_P}{K_S}\frac{1}{\gamma_{r}}$ 
the maximizers of $h_i$ are
$$
\tilde{\lambda}_i(\gamma)= 
\frac{K_S \gamma \gamma_i}{K_P}>1, \quad i=1, \ldots, r,
$$
and the remaining maximizers are $\tilde{\lambda}_i(\gamma)=1$, $i=r+1, \ldots, N$
(as a matter of fact, $\gamma_{K_P+1}=\cdots=\gamma_{N}=0$).
Thus, the
compressed likelihood under $H_1$ becomes
\begin{align}
\nonumber
&L_1(\widehat{\bR}, \widehat{\tilde{\bR}}_s, \gamma; \bZ)
=-NK \log \pi -NK_S \log \gamma 
\\ \nonumber &- 2K \log |\det(\bK)| +
 \sum_{i={1}}^{r} \left[ K \log \frac{K_P}{\gamma_i}+ K_S \log \frac{K_S \gamma \gamma_i}{K_P} \right]
\\ &+
 \sum_{i={r+1}}^{N}K\log \frac{\gamma K}{\gamma \gamma_{i} + 1}
-NK.
\label{eq:log_likelihood_H1_7_known_rank-4}
\end{align}
If instead $K_P=r+m$, $m \geq 1$, we have to distinguish the following cases
\begin{itemize}
\item
if $\frac{K_P}{K_S}\frac{1}{\gamma_{r}} \leq \gamma < \frac{K_P}{K_S}\frac{1}{\gamma_{r+1}}$
the maximizers of $h_i$ are
$$
\tilde{\lambda}_i(\gamma)= 
\frac{K_S \gamma \gamma_i}{K_P}>1, \quad i=1, \ldots, r,
$$
and the remaining maximizers are $\tilde{\lambda}_i(\gamma)=1$, $i=r+1, \ldots, N$.
Thus, the
compressed likelihood under $H_1$ becomes
\begin{align}
\nonumber
&L_1(\widehat{\bR}, \widehat{\tilde{\bR}}_s, \gamma; \bZ)
=-NK \log \pi -NK_S \log \gamma 
\\ \nonumber &- 2K \log |\det(\bK)| 
\\ \nonumber &+
\sum_{i={1}}^{r} \left[ K \log \frac{K_P}{\gamma_i}+ K_S \log \frac{K_S \gamma \gamma_i}{K_P} \right]
\\ &+
 \sum_{i={r+1}}^{N} K \log \frac{\gamma K}{\gamma \gamma_{i} + 1}
-NK.
\label{eq:log_likelihood_H1_7_known_rank-5}
\end{align}
\item
if $
\frac{K_P}{K_S}\frac{1}{\gamma_{r+h}}
\leq \gamma < \frac{K_P}{K_S}\frac{1}{\gamma_{r+h+1}}$, $h=1, \ldots, m-1$, 
the maximizers of $h_i$ are (this case refers to $m>1$)
$$
\tilde{\lambda}_i(\gamma)= 
\frac{K_S \gamma \gamma_i}{K_P}>1, \quad i=1, \ldots, r+h,
$$
and the remaining maximizers are $\tilde{\lambda}_i(\gamma)=1$, $i=r+h+1, \ldots, N$.
However, the fact that the maximum number of $\lambda_i >1$ has to be $r$ (at most) together with the descending order of the $\lambda_i$s implies that
the
compressed likelihood under $H_1$ is still given by eq. (\ref{eq:log_likelihood_H1_7_known_rank-5}).
\item
Finally,
if $\gamma \geq \frac{K_P}{K_S}\frac{1}{\gamma_{r+m}}$ 
the maximizers of $h_i$ are
$$
\tilde{\lambda}_i(\gamma)= 
\frac{K_S \gamma \gamma_i}{K_P}>1, \quad i=1, \ldots, r+m,
$$
and the remaining maximizers (if any) are $\tilde{\lambda}_i(\gamma)=1$, $i=r+m+1, \ldots, N$
(as a matter of fact, $\gamma_{K_P+1}=\cdots=\gamma_{N}=0$).
However, the
compressed likelihood under $H_1$ 
is still given by eq. (\ref{eq:log_likelihood_H1_7_known_rank-5}).
\end{itemize}
\end{itemize}
The statement of the theorem follows.

\section{Proof of Theorem \ref{Theorem4}}
\label{App:ProofTheorem4}
Preliminary let
\be
\nonumber
g( \gamma)
=-NK_S \log \gamma 
+ g_1(\gamma)
\ee
with 
$$
g_1(\gamma)=
\begin{cases}
\begin{cases}
\sum_{i=1}^{r-1} K \log \frac{\gamma K}{\gamma \gamma_i + \widehat{\lambda}_i(\gamma)}+  \sum_{i=1}^{r-1} K_S \log \widehat{\lambda}_i (\gamma)
\\ 
+ \sum_{i={r}}^{N} K \log \frac{\gamma K}{\gamma \gamma_{i} + 1}
\\ 
\gamma < \frac{K_P}{K_S} \frac{1}{\gamma_r}
\end{cases}  \\
\begin{cases}
\sum_{i=1}^r \left[ K \log \frac{K_P}{\gamma_i}+K_S \log \frac{K_S \gamma \gamma_i}{K_P} \right]
\\ + \sum_{i=r+1}^N K \log \frac{\gamma K}{\gamma \gamma_i+1} \\
\gamma \geq \frac{K_P}{K_S} \frac{1}{\gamma_r}
\end{cases}
\end{cases}
$$
More specifically, we have that
\be
g_1(\gamma)=
\begin{cases}
\begin{cases}
K\sum_{i=1}^N\log \frac{\gamma K}{\gamma \gamma_i + 1}, \\
\gamma < \frac{K_P}{K_S} \frac{1}{\gamma_1},
\end{cases} \\
\begin{cases}
\sum_{i=1}^{j-1} \left[ K \log \frac{K_P}{\gamma_i}
+ K_S \log \frac{K_S \gamma \gamma_i}{K_P} \right] \\
+K\sum_{i=j}^{N} \log \frac{\gamma K}{\gamma\gamma_i+1}, \\
\frac{K_P}{K_S} \frac{1}{\gamma_{j-1}} \leq \gamma < \frac{K_P}{K_S} \frac{1}{\gamma_j}, j=2, \ldots, r,
\end{cases} \\
\begin{cases}
\sum_{i=1}^r \left[ K \log \frac{K_P}{\gamma_i}+K_S \log \frac{K_S \gamma \gamma_i}{K_P} \right]
\\ + \sum_{i=r+1}^N K \log \frac{\gamma K}{\gamma \gamma_i+1}, \\
\gamma \geq \frac{K_P}{K_S} \frac{1}{\gamma_r}.
\end{cases}
\end{cases}
\label{eq:g1_gamma}
\ee
Then, notice that
$$
\lim_{\gamma \rightarrow 0} g(\gamma)=
\lim_{\gamma \rightarrow 0} \left[ -NK_S \log \gamma
+ K\sum_{i=1}^N\log \frac{\gamma K}{\gamma \gamma_i + 1}
\right]=-\infty
$$
and
\begin{align*}
&\lim_{\gamma \rightarrow +\infty} g(\gamma)=
\lim_{\gamma \rightarrow +\infty} \left[ -NK_S \log \gamma
+
\sum_{i=1}^r K_S \log \frac{K_S \gamma \gamma_i}{K_P} \right.
\\ &+ \left. \sum_{i=r+1}^N K \log \frac{\gamma K}{\gamma \gamma_i+1}
\right] + \sum_{i=1}^r K \log \frac{K_P}{\gamma_i}
\\ &=
\lim_{\gamma \rightarrow +\infty} \left[ (r-N)K_S \log \gamma
+(N-r)K \log \gamma \right.
\\ &+ \left. \sum_{i=r+1}^N K \log \frac{K}{\gamma \gamma_i+1}
\right] + \sum_{i=1}^r K \log \frac{K_P}{\gamma_i}
+\sum_{i=1}^r K_S \log \frac{K_S \gamma_i}{K_P}
\\ &=
\lim_{\gamma \rightarrow +\infty} \left[ K_P(N-r) \log \gamma
-K\sum_{i=r+1}^{K_P} \log \left( \gamma \gamma_{i}+1 \right)
\right]
\\ &
+ \sum_{i=1}^r K \log \frac{K_P}{\gamma_i}
+\sum_{i=1}^r K_S \log \frac{K_S \gamma_i}{K_P}
+ \sum_{i=r+1}^N K \log K
=
-\infty
\end{align*}
exploiting $\gamma_i \neq 0$, $i=r+1, \ldots, K_P$
and provided that $(K_P-r)K > (N-r)K_P$ or, equivalently,
$(K_P-r)K_S > (N-r)K_P-(K_P-r)K_P=(N-K_P)K_P$.

Thus, the maximum corresponds to a stationary point. To compute the stationary points we observe that
$$
\frac{dg_1}{d\gamma}(\gamma)=
\begin{cases}
\begin{cases}
\sum_{i=1}^N \frac{K}{\gamma(\gamma \gamma_i + 1)} \\
\gamma < \frac{K_P}{K_S} \frac{1}{\gamma_1}
\end{cases} \\
\begin{cases}
\sum_{i=1}^{j-1} \frac{K_S}{\gamma}
+\sum_{i=j}^{N} \frac{K}{\gamma(\gamma\gamma_i+1)} \\
\frac{K_P}{K_S} \frac{1}{\gamma_{j-1}} \leq \gamma < \frac{K_P}{K_S} \frac{1}{\gamma_j}, j=2, \ldots, r
\end{cases} \\
\begin{cases}
\frac{rK_S}{\gamma} + \sum_{i=r+1}^N \frac{K}{\gamma(\gamma \gamma_i+1)} \\
\gamma \geq \frac{K_P}{K_S} \frac{1}{\gamma_r}
\end{cases}
\end{cases}
$$
It is easy to check that 
$$
\frac{dg}{d\gamma}(\gamma)=-\frac{NK_S}{\gamma} + \frac{dg_1}{d\gamma}(\gamma) >0, \quad \gamma < \frac{K_P}{K_S} \frac{1}{\gamma_1}.
$$
In fact, $\gamma \gamma_i < \frac{K_P}{K_S}$ implies
$$
\frac{dg_1}{d\gamma}(\gamma)=\sum_{i=1}^N \frac{K}{\gamma(\gamma \gamma_i + 1)}
> 
\frac{1}{\gamma} \sum_{i=1}^N K_S=\frac{NK_S}{\gamma}.
$$
Similarly, $\gamma \gamma_i < \frac{K_P}{K_S}$, $i=j, \ldots, N,$
implies that
$$
\sum_{i=j}^N \frac{K}{\gamma(\gamma \gamma_i + 1)}
> 
\frac{1}{\gamma} \sum_{i=j}^N K_S= \frac{1}{\gamma} K_S(N-j+1) 
$$
and eventually 
$$
\frac{dg}{d\gamma}(\gamma)>0, \quad 
\frac{K_P}{K_S} \frac{1}{\gamma_{j-1}} \leq \gamma < \frac{K_P}{K_S} \frac{1}{\gamma_j}, j=2, \ldots, r.
$$
Moreover, for $\gamma \geq \frac{K_P}{K_S} \frac{1}{\gamma_r}$, we have that 
\begin{align*}
&\frac{dg}{d\gamma}(\gamma)=-\frac{NK_S}{\gamma} + \frac{rK_S}{\gamma} + \sum_{i=r+1}^N \frac{K}{\gamma(\gamma \gamma_i+1)}
\\ &=
-\frac{(N-r)K_S}{\gamma}+ \sum_{i=r+1}^{K_P} \frac{K}{\gamma(\gamma \gamma_i+1)}+ \frac{(N-K_P)K}{\gamma}
\\ & \geq
-\frac{(N-r)K_S}{\gamma}+ \frac{K(K_P-r)}{\gamma(\gamma \gamma_r+1)}+ \frac{(N-K_P)K}{\gamma}
\end{align*}
and, in particular, $\frac{dg}{d\gamma}(\gamma)$ is non-negative at
$\gamma=\frac{K_P}{K_S} \frac{1}{\gamma_r}$.
Thus, since
\begin{align*}
&\lim_{\gamma \rightarrow + \infty} \gamma \frac{dg}{d\gamma}(\gamma)=-(N-r)K_S+(N-K_P)K
\\ &=-(N-r)K_S+(N-K_P)K_S+(N-K_P)K_P
\\ &=
K_S(r-K_P)+(N-K_P)K_P<0,
\end{align*}
it turns out that the maximum is attained at the (unique\footnote{The solution is apparently unique since the left-hand side of the equation is a stricly decreasing function of $\gamma>0$.}) solution of the equation
\begin{align*}
&\sum_{i=r+1}^{K_P} \frac{K}{(\gamma \gamma_i+1)}=- (N-K_P)K+(N-r)K_S
\\ &=(K_P-r)K_S- (N-K_P)K_P.
\end{align*}
 
\section{Proof of Theorem \ref{Theorem5}}
\label{App:ProofTheorem5}
Preliminary observe that $K_P \geq r$ implies that the rank of the matrix 
$\tilde{\bS}_S^{-1/2} \tilde{\bS}_P \tilde{\bS}_S^{-1/2}
\in\C^{r\times r}$
is $r$ and hence that $\gamma_r \neq 0$ (with probability one).
Then, define
$$
f(\gamma)=-rK_S \log \gamma+f_1(\gamma)+f_2(\gamma)
$$
where
$$
f_1(\gamma)=
K\sum_{i=1}^r\log \frac{\gamma K}{\gamma \gamma_i + \widehat{\lambda}_i(\gamma)}
+ K_S \sum_{i=1}^r\log \widehat{\lambda}_i (\gamma)
$$
and
$$
f_2(\gamma)=
-K \sum_{i=1}^{N-r} \log \left( \frac{1}{\gamma} + \delta_i \right)
-(N-r)K_S \log \gamma.
$$
Thus, supposing that $K_P \geq r$ and, hence, $\gamma_r \neq 0$, we have that
$$
f_1( \gamma)
=
\begin{cases}
\begin{cases}
K\sum_{i=1}^r\log \frac{\gamma K}{\gamma \gamma_i+1},
\\
\gamma < \frac{K_P}{K_S} \frac{1}{\gamma_{1}} 
\end{cases}
\\
\begin{cases}
K\sum_{i=1}^j\log \frac{K_P}{\gamma_i}
+ K\sum_{i=j+1}^r\log \frac{\gamma K}{\gamma \gamma_i+1}  \\
+K_S\sum_{i=1}^j\log \left( \frac{K_S}{K_P} \gamma \gamma_i \right), \\
\frac{K_P}{K_S} \frac{1}{\gamma_{j}} \leq
\gamma < \frac{K_P}{K_S} \frac{1}{\gamma_{j+1}}, 1 \leq j \leq r-1
\end{cases}
\\
\begin{cases}
K\sum_{i=1}^r\log \frac{K_P}{\gamma_i}
+K_S\sum_{i=1}^r\log \left( \frac{K_S}{K_P} \gamma \gamma_i \right), \\
\gamma \geq \frac{K_P}{K_S} \frac{1}{\gamma_{r}}
\end{cases}
\end{cases}
$$
As to $f_2(\gamma)$, it can be written as
\begin{align*}
f_2(\gamma) &=-K \sum_{i=1}^{\min(K_P,N-r)} \log \left(
\frac{1+\gamma \delta_i}{\gamma} \right)
\\ &+\left[ K \left( N-r-\min(K_P,N-r)\right) -(N-r) K_S \right] \log \gamma.
\end{align*}
Thus, it is easy to check that $\lim_{\gamma \rightarrow 0} f(\gamma)= -\infty$
and $\lim_{\gamma \rightarrow +\infty} f(\gamma)= -\infty$.
It follows that the global maximum is achieved at a stationary point.
\medskip

 Moreover,
the derivative of $f_1^{\prime}(\gamma)=-rK_S \log \gamma+f_1(\gamma)$ with respect to $\gamma$ is given by
$$
\frac{d f_1^{\prime}(\gamma)}{d\gamma}=
\begin{cases}
\begin{cases}
-\frac{rK_S}{\gamma}+\frac{K}{\gamma}\sum_{i=1}^r \frac{1}{\gamma \gamma_i+1},
\\
\gamma < \frac{K_P}{K_S} \frac{1}{\gamma_{1}}
\end{cases} \\
\begin{cases}
-\frac{(r-j)K_S}{\gamma}
+ \frac{K}{\gamma} \sum_{i=j+1}^r \frac{1}{\gamma \gamma_i+1}, \\
\frac{K_P}{K_S} \frac{1}{\gamma_{j}} \leq
\gamma < \frac{K_P}{K_S} \frac{1}{\gamma_{j+1}}, 1 \leq j \leq r-1
\end{cases}
\\
\begin{cases}
0, \\
\gamma \geq \frac{K_P}{K_S} \frac{1}{\gamma_{r}}
\end{cases}
\end{cases}
$$
while the derivative of $f_2(\gamma)$ can be written as
\begin{align*}
\frac{d f_2(\gamma)}{d\gamma} &=
\frac{1}{\gamma} \left[ K \sum_{i=1}^{\min(K_P,N-r)}  
\frac{1}{1+\gamma \delta_i} \right.
\\ &\left. +\left( K \left( N-r-\min(K_P,N-r)\right) -(N-r) K_S \right) \right].
\end{align*}
\medskip

Notice that $\frac{d f_1^{\prime}(\gamma)}{d\gamma}$ is positive if $\gamma < \frac{K_P}{K_S} \frac{1}{\gamma_r}$ and is equal to zero if $\gamma \geq \frac{K_P}{K_S} \frac{1}{\gamma_r}$.
As to the derivative of $f_2(\gamma)$, it is apparent that 
it is positive and strictly decreasing up to the unique value of $\gamma$,
say $\gamma^*$, that solves the equation
\begin{align*}
&\sum_{i=1}^{\min(K_P,N-r)}  
\frac{1}{1+\gamma \delta_i}
+ \left( N-r-\min(K_P,N-r)\right) 
\\ &-\frac{(N-r) K_S}{K} =0,
\end{align*}
while it is negative if $\gamma > \gamma^*$.
In fact, the function $\sum_{i=1}^{\min(K_P,N-r)}  
\frac{1}{1+\gamma \delta_i}$ is positive and strictly decreasing; moreover, 
$$
\lim_{\gamma \rightarrow 0} K \sum_{i=1}^{\min(K_P,N-r)}  
\frac{1}{1+\gamma \delta_i}=K \min(K_P,N-r)
$$
and
$$\min(K_P,N-r) 
+\frac{(N-r)K_P}{K}-\min(K_P,N-r) >0,
$$
but
$$
\lim_{\gamma \rightarrow +\infty} K \sum_{i=1}^{\min(K_P,N-r)}  
\frac{1}{1+\gamma \delta_i}=0
$$
and $K \left( N-r-\min(K_P,N-r)\right) -(N-r) K_S<0$
(when $K_P \leq N-r$ recall that $K_S\geq N$).

As a consequence we conclude that 
if $\gamma^* \geq \frac{K_P}{K_S} \frac{1}{\gamma_r}$
the derivative of $f(\gamma)$ has a unique zero at $\gamma^*$;  otherwise,
the stationary points of $f$ and hence its global maximum  belong to the interval $(\gamma^*, \frac{K_P}{K_S} \frac{1}{\gamma_r})$. In fact,
$$
\frac{d f}{d\gamma}(\gamma^*)=\frac{d f_1^{\prime}}{d\gamma}(\gamma^*)+\frac{d f_2}{d\gamma}(\gamma^*)>0
$$
and
$$
\frac{d f}{d\gamma}\left(\frac{K_P}{K_S} \frac{1}{\gamma_r}\right)=\frac{d f_1^{\prime}}{d\gamma}\left(\frac{K_P}{K_S} \frac{1}{\gamma_r}\right)+\frac{d f_2}{d\gamma}\left(\frac{K_P}{K_S} \frac{1}{\gamma_r}\right)<0.
$$

\bibliographystyle{IEEEtran}
\bibliography{GLRTs_v12.bib}

\begin{thebibliography}{10}
\providecommand{\url}[1]{#1}
\csname url@samestyle\endcsname
\providecommand{\newblock}{\relax}
\providecommand{\bibinfo}[2]{#2}
\providecommand{\BIBentrySTDinterwordspacing}{\spaceskip=0pt\relax}
\providecommand{\BIBentryALTinterwordstretchfactor}{4}
\providecommand{\BIBentryALTinterwordspacing}{\spaceskip=\fontdimen2\font plus
\BIBentryALTinterwordstretchfactor\fontdimen3\font minus
  \fontdimen4\font\relax}
\providecommand{\BIBforeignlanguage}[2]{{%
\expandafter\ifx\csname l@#1\endcsname\relax
\typeout{** WARNING: IEEEtran.bst: No hyphenation pattern has been}%
\typeout{** loaded for the language `#1'. Using the pattern for}%
\typeout{** the default language instead.}%
\else
\language=\csname l@#1\endcsname
\fi
#2}}
\providecommand{\BIBdecl}{\relax}
\BIBdecl

\bibitem{Kelly86}
E.~J. Kelly, ``{An adaptive detection algorithm},'' \emph{IEEE Transactions on
  Aerospace and Electronic Systems}, no.~2, pp. 115--127, 1986.

\bibitem{Kelly-Forsythe}
E.~J. Kelly and K.~Forsythe, ``{Adaptive Detection and Parameter Estimation for
  Multidimensional Signal Models},'' Lincoln Lab, MIT, Lexington, US, Technical
  Report 848, 1989.

\bibitem{Chen-Reed}
W.-S. Chen and I.~S. Reed, ``A new cfar detection test for radar,''
  \emph{Digital Signal Processing}, vol.~1, no.~4, pp. 198--214, 1991.

\bibitem{Robey}
F.~Robey, D.~Fuhrmann, E.~Kelly, and R.~Nitzberg, ``A cfar adaptive matched
  filter detector,'' \emph{IEEE Transactions on Aerospace and Electronic
  Systems}, vol.~28, no.~1, pp. 208--216, 1992.

\bibitem{CLR1995}
E.~{Conte}, M.~{Lops}, and G.~{Ricci}, ``Asymptotically optimum radar detection
  in compound-gaussian clutter,'' \emph{IEEE Transactions on Aerospace and
  Electronic Systems}, vol.~31, no.~2, pp. 617--625, 1995.

\bibitem{Scharf-McWhorter}
L.~Scharf and L.~McWhorter, ``Adaptive matched subspace detectors and adaptive
  coherence estimators,'' in \emph{Conference Record of The Thirtieth Asilomar
  Conference on Signals, Systems and Computers}, vol.~2, 1996, pp. 1114--1117.

\bibitem{Scharf-book}
L.~L. Scharf, \emph{Statistical Signal Processing: Detection, Estimation, and
  Time Series Analysis}.\hskip 1em plus 0.5em minus 0.4em\relax Addison-Wesley
  Publishing Company, 1991.

\bibitem{Scharf-Friedlander1994}
L.~L. {Scharf} and B.~{Friedlander}, ``Matched subspace detectors,'' \emph{IEEE
  Transactions on Signal Processing}, vol.~42, no.~8, pp. 2146--2157, 1994.

\bibitem{Kraut-Scharf1999}
S.~{Kraut} and L.~L. {Scharf}, ``The {CFAR} adaptive subspace detector is a
  scale-invariant {GLRT},'' \emph{IEEE Transactions on Signal Processing},
  vol.~47, no.~9, pp. 2538--2541, 1999.

\bibitem{Kraut-Scharf-McWhorter}
S.~Kraut, L.~L. Scharf, and L.~T. McWhorter, ``Adaptive subspace detectors,''
  \emph{IEEE Transactions on Signal Processing}, vol.~49, no.~1, pp. 1--16,
  January 2001.

\bibitem{Bose-Steinhardt}
S.~{Bose} and A.~O. {Steinhardt}, ``A maximal invariant framework for adaptive
  detection with structured and unstructured covariance matrices,'' \emph{IEEE
  Transactions on Signal Processing}, vol.~43, no.~9, pp. 2164--2175, 1995.

\bibitem{Gerlach1997}
K.~{Gerlach}, M.~{Steiner}, and F.~C. {Lin}, ``Detection of a spatially
  distributed target in white noise,'' \emph{IEEE Signal Processing Letters},
  vol.~4, no.~7, pp. 198--200, 1997.

\bibitem{CDMR2001}
E.~{Conte}, A.~{De Maio}, and G.~{Ricci}, ``{GLRT}-based adaptive detection
  algorithms for range-spread targets,'' \emph{IEEE Transactions on Signal
  Processing}, vol.~49, no.~7, pp. 1336--1348, 2001.

\bibitem{Gini-Farina2002}
F.~{Gini} and A.~{Farina}, ``Vector subspace detection in compound-gaussian
  clutter. {Part I}: survey and new results,'' \emph{IEEE Transactions on
  Aerospace and Electronic Systems}, vol.~38, no.~4, pp. 1295--1311, 2002.

\bibitem{BDMGR}
F.~Bandiera, A.~De~Maio, A.~S. Greco, and G.~Ricci, ``{Adaptive Radar Detection
  of Distributed Targets in Homogeneous and Partially Homogeneous Noise Plus
  Subspace Interference},'' \emph{IEEE Transactions on Signal Processing},
  vol.~55, no.~4, pp. 1223--1237, 2007.

\bibitem{BBORS2007}
F.~{Bandiera}, O.~{Besson}, D.~{Orlando}, G.~{Ricci}, and L.~L. {Scharf},
  ``{GLRT}-based direction detectors in homogeneous noise and subspace
  interference,'' \emph{IEEE Transactions on Signal Processing}, vol.~55,
  no.~6, pp. 2386--2394, 2007.

\bibitem{CDMO2016_1}
D.~{Ciuonzo}, A.~{De Maio}, and D.~{Orlando}, ``A unifying framework for
  adaptive radar detection in homogeneous plus structured interference— {Part
  I}: On the maximal invariant statistic,'' \emph{IEEE Transactions on Signal
  Processing}, vol.~64, no.~11, pp. 2894--2906, 2016.

\bibitem{CDMO2016_2}
------, ``A unifying framework for adaptive radar detection in homogeneous plus
  structured interference— {Part II}: Detectors design,'' \emph{IEEE
  Transactions on Signal Processing}, vol.~64, no.~11, pp. 2907--2919, 2016.

\bibitem{CFR2020}
A.~Coluccia, A.~Fascista, and G.~Ricci, ``A novel approach to robust radar
  detection of range-spread targets,'' \emph{Signal Processing}, vol. 166, p.
  107223, 2020.

\bibitem{Ricci-Scharf}
G.~Ricci and L.~L. Scharf, ``Adaptive radar detection of extended gaussian
  targets,'' in \emph{12nd MIT Lincoln Labs Workshop on Adaptive Sensor and
  Array Processing}, Lexington, MA (USA), March 16-18 2004.

\bibitem{BCCRV}
O.~Besson, A.~Coluccia, E.~Chaumette, G.~Ricci, and F.~Vincent, ``Generalized
  likelihood ratio test for detection of gaussian rank-one signals in gaussian
  noise with unknown statistics,'' \emph{IEEE Transactions on Signal
  Processing}, vol.~65, no.~4, pp. 1082--1092, February 15 2017.

\bibitem{Steedly-Moses}
W.~M. Steedly and R.~L. Moses, ``High resolution exponential modeling of fully
  polarized radar returns,'' \emph{IEEE Trans. on Aerospace and Electronics
  Systems}, vol.~27, no.~3, pp. 459--469, May 1991.

\bibitem{VanTrees-book1}
H.~L. Van~Trees, \emph{{Detection, Estimation, and Modulation Theory, Part I:
  Detection, Estimation, and Linear Modulation Theory}}.\hskip 1em plus 0.5em
  minus 0.4em\relax John Wiley \& Sons, 2004.

\bibitem{Jin-Friedlander}
Y.~Jin and B.~Friedlander, ``A {CFAR} adaptive subspace detector for
  second-order gaussian signals,'' \emph{IEEE Transactions on Signal
  Processing}, vol.~53, no.~3, pp. 871--884, March 2005.

\bibitem{RichardsBasicPrinciples}
M.~A. Richards, W.~A. Holm, and J.~Scheer, \emph{Principles of Modern Radar:
  Basic Principles}, ser. Electromagnetics and Radar.\hskip 1em plus 0.5em
  minus 0.4em\relax Institution of Engineering and Technology, 2010.

\bibitem{Muirhead}
R.~J. Muirhead, \emph{{Aspects of multivariate statistical theory}}.\hskip 1em
  plus 0.5em minus 0.4em\relax John Wiley \& Sons, 1982.

\bibitem{MatrixAnalysis}
R.~A. Horn and C.~R. Johnson, \emph{Matrix Analysis}.\hskip 1em plus 0.5em
  minus 0.4em\relax Cambridge University Press, 1985.

\bibitem{mirsky1959trace}
L.~Mirsky, ``On the trace of matrix products,'' \emph{Mathematische
  Nachrichten}, vol.~20, pp. 171--174, 1959.

\bibitem{Handbook_of_matrices}
H.~Lütkepohl, \emph{Handbook of Matrices}.\hskip 1em plus 0.5em minus
  0.4em\relax John Wiley \& Sons, 1996.

\bibitem{Stoica_alternating}
P.~Stoica and Y.~Selen, ``{Cyclic minimizers, majorization techniques, and the
  expectation-maximization algorithm: a refresher},'' \emph{IEEE Signal
  Processing Magazine}, vol.~21, no.~1, pp. 112--114, 2004.

\end{thebibliography}

\end{document}